\documentclass{article}
\usepackage{graphicx}
\usepackage{color}
\usepackage{xcolor}
\usepackage{natbib}
\usepackage{comment}
\usepackage{bbm}
\usepackage{booktabs}
\usepackage{makecell}
\usepackage{subcaption}  
\usepackage{threeparttable} 
\usepackage{float}
\usepackage{multirow}

\usepackage[margin=1in]{geometry}
\setcitestyle{authoryear, open={((},close={)}}
\usepackage{amssymb,amsmath,amsthm}

\newcommand{\bsX}{\boldsymbol X}
\newcommand{\bsZ}{\boldsymbol Z}
\newcommand{\bsU}{\boldsymbol U}
\newcommand{\bsY}{\boldsymbol Y}
\newcommand{\bsW}{\boldsymbol W}
\newcommand{\I}{\mathbbm 1}

\newcommand{\bsmu}{\boldsymbol \mu}
\newcommand{\bsSigma}{\boldsymbol \Sigma}
\newcommand{\bsSigmaN}{\boldsymbol {\Sigma_{\textbf{NA}}}}
\newcommand{\bsOmega}{\boldsymbol \Omega}

\newcommand{\bstheta}{\boldsymbol \theta}
\newcommand{\bstau}{\boldsymbol \tau}
\newcommand{\bspsi}{\boldsymbol \psi}
\newcommand{\hbstau}{\boldsymbol{\hat \tau}}

\newcommand{\bsSig}{\boldsymbol  \Sigma}

\newcommand{\bsnu}{\boldsymbol \nu_{\textbf{NA}}}

\newcommand{\bsn}{\boldsymbol n}
\newcommand{\bsns}{\boldsymbol{n^\star}}
\newcommand{\bsnd}{\boldsymbol{n^\dag}}
\newcommand{\bsd}{\boldsymbol{d}}
\newcommand{\bsdd}{\boldsymbol{d^\dag}}

\newcommand{\bsdeltad}{\boldsymbol{\hat \delta_{\textbf{D}}}}
\newcommand{\btaujs}{\boldsymbol \delta_{\textbf{D}}}

\newtheorem{lemma}{Lemma}

\newtheorem{claim}{Claim}

\newenvironment{proof*}
  {\proof}
  {\endproof}

\def\E{\mathbb{E}}
\def\e{\mathbb{E}}
\def\var{\text{var}}
\def\cov{\text{cov}}

\newcommand{\Tr}{\text{tr}}

\newcommand{\hbsSig}{ \boldsymbol{\hat \Sigma}}
\newcommand{\ident}{\boldsymbol I}
\newcommand{\bsdelt}{\boldsymbol{\hat \delta}}
\newcommand{\bsdeltb}{\boldsymbol {\hat \delta_B}}
\newcommand{\bsdelts}{\boldsymbol {\hat \delta_S}}

\newcommand{\tran}{\mathsf{T}}

\newcommand{\lambdamax}{\lambda_{\text{max}}(\bsSig)}
\newcommand{\lambdamaxn}{\lambda_{\text{max}}(\bsSigmaN)}

\newcommand{\sna}{S_{\text{NA}}}

\definecolor{todo}{RGB}{210, 43, 43}

\DeclareMathOperator*{\argmin}{arg\,min}

\title{Adaptive Experimental Design Using Shrinkage Estimators}
\author{
Evan T.\ R.\ Rosenman\\
\small Department of Mathematical Sciences\\
\small Claremont McKenna College \\
\small Claremont, CA, United States
\and
Kristen B. Hunter\\
\small School of Mathematics and Statistics\\
\small University of New South Wales \\
\small Sydney, New South Wales, Australia 
}

\date{\today}                   

\begin{document}

\maketitle

\abstract{In multi-armed trials, adaptive designs are a popular way to increase estimation efficiency or identify optimal treatments. Several recent papers have proposed adaptive variants of the classical Neyman allocation to assign treatments in sequential trials, with the goal of minimizing the error of a Horvitz-Thompson-style estimator. However, this approach may be inefficient, because it fails to borrow information across the treatment arms. 

In this paper, we consider adaptivity in a sequential trial with $K$ active treatments and a control, and suggest the use of Stein-like shrinkage estimators to obtain the final causal estimates. 
These estimators share information across arms, yielding provable reductions in expected squared error loss relative to estimating each causal effect in isolation. Moreover, for each of our candidate shrinkers, the risk is the expectation of ratios of Gaussian quadratic forms, and can be computed efficiently via numerical integration. Hence, we suggest a simple algorithm for sequential adaptivity: assign treatments to each new arrival by choosing the arm that will minimize the estimated shrinker loss.

Through simulations, we demonstrate that this approach can yield meaningful reductions in estimation error, especially in the low signal-to-noise regime. We also characterize how our adaptive algorithm assigns treatments differently than would a sequential Neyman allocation, and suggest a method for constructing shorter confidence intervals at the trial's conclusion.
}

\tableofcontents

\section{Introduction}\label{sec:intro}

Adaptive designs offer a powerful tool to improve experimentation. In these designs, the treatment assignment for each new arrival is updated based on outcomes previously observed in the trial. These updates can be targeted towards several goals, including efficient estimation of causal effects \citep{dai2023clip}, identification of optimal treatments \citep{villar2015multi}, or maximizing the utility of participants in the trial \citep{rosenberger2001optimal}. 


We consider the first goal -- improving estimation efficiency -- in the context of a multi-armed randomized experiment in which there are $K$ active treatment levels, as well as a control level. We suppose the trial is ``online'' -- i.e., individuals arrive sequentially and, upon arrival, must be assigned a treatment $k \in \{0, \dots, K\}$, where $k = 0$ denotes the control and $k = 1, \dots, K$ denote the active treatments. We assume that trial participants are randomly sampled from an infinite super-population.
For each individual in this super-population, we define potential outcomes $Y(k)$ corresponding to every possible treatment level \citep{rubin1974estimating}. The causal quantities of interest are the average treatment effects
\[
\tau_k = \E\left( Y(k) - Y(0) \right), \hspace{5mm} k = 1, \dots, K,
\]
where the expectation is taken with respect to the super-population distribution of potential outcomes.
We collect these parameters into a vector $\bstau = (\tau_1, \dots, \tau_K) \in \mathbb{R}^K$.

A central challenge in adaptive experiments is the design of treatment assignment rules to exploit the accumulating data. The Neyman allocation \citep{Neyman1934}, a classical sampling result, suggests allocating units to treatments in proportion to their potential outcome standard deviations. Recent work has extended this idea to sequential and adaptive settings \citep{dai2023clip, zhao2023adaptive, chen2025sigmoid}, proposing methods to learn these standard deviations and update 
treatment probabilities as outcome data accrue. These approaches are typically paired with Horvitz-Thompson-style estimators to retain unbiasedness under adaptive assignment.


Such methods may not be optimal when the goal is to estimate \emph{all} of the treatment effects simultaneously. In many applications, the treatment effects are related, making it possible to borrow strength across estimates. Shrinkage estimators provide a principled way to exploit this structure. They can achieve lower expected squared error loss than estimating each effect separately, even in heteroscedastic settings.

In this manuscript, we study adaptive experimental design when causal estimation is conducted using a Stein-like shrinkage estimator for the vector of treatment effects. Under standard approximate normal distributions for arm-level estimators, the expected squared error loss of 
candidate shrinkers can be written as an expectation of ratios of Gaussian quadratic forms. This representation allows the risk to be computed efficiently via numerical integration, making it feasible to design adaptive assignment rules that target minimization of the shrinkage risk. 

The remainder of this paper proceeds as follows. Section \ref{sec:litreview} reviews the relevant literature on Neyman allocation and adaptive designs. Section \ref{sec:theory} introduces three candidate shrinkage estimators and explains how their risk can be estimated using an adaptation of Stein's Unbiased Risk Estimate \citep{stein1981estimation}. Section \ref{sec:oracle} considers oracle designs in the case where the potential outcome means and variances are known for each treatment arm. We begin by solving for the Neyman allocation, and then consider the risk-minimizing allocations for each of the candidate shrinkers. Section \ref{sec:adaptivity} introduces our greedy adaptivity scheme, demonstrating how shrinker risk can be efficiently computed via numerical integration. Subsequent simulations evince the benefits of adaptivity in practice. Section \ref{sec:cis} discusses the construction of confidence intervals at the end of the adaptive trial. Section \ref{sec:conclusion} concludes.

\section{Literature Review}\label{sec:litreview}

Adaptive experimental designs can vary widely in their objectives. One fundamental distinction is between designs that prioritize identifying and exploiting the best treatment and those that prioritize precise estimation of causal effects. Another is the tradeoff between exploitation and exploration. 


In a multi-armed bandit setting, the primary goal is to maximize the expected reward by quickly identifying the best-performing arm, so that the largest possible number of units are assigned to it.  A canonical approach is Thompson sampling, which randomly assigns units to arms proportional to the posterior probability that each arm is optimal \citep{Thompson1933, AgrawalGoyal2012, Russo2018}. More recent work has considered a ``top-two" variant of Thompson sampling, where more aggressive exploration is induced by randomizing within the two arms that appear most promising based on the posterior \citep{russo2016simple, kasy2021adaptive, rosenzweig2022conversations}. \citet{dimmery2019shrinkage} shows that Empirical Bayes shrinkage can be applied to produce results similar to top-two Thompson sampling, but with better performance in selecting a set of high-performing arms rather than solely the top two.

Unlike traditional randomized trials, bandit designs typically do not include a control arm, and the emphasis is not on precise estimation of effect sizes. Instead, the aim is to minimize cumulative regret by rapidly exploiting the reward-maximizing treatment option. As a simple example, consider an A/B test in which a company seeks to determine which color of a website search bar -- red or blue -- maximizes revenue. The company does not care about pairwise comparisons between colors, nor about estimating revenue levels with high precision. Even if the 
premium of red over blue is small, red would be chosen as it yields the highest expected profit. Accordingly, the firm wishes to minimize exposure to inferior designs during the experiment, rather than to learn fine-grained effect sizes.

By contrast, efficiency-oriented adaptive designs aim to estimate causal effects as precisely as possible. The treatment allocation is chosen to minimize the total variance of treatment effect estimators, often under a fixed sample size constraint. In this setting, trials typically include a control arm, and treatment effects are defined as contrasts relative to that control. The success of the design is evaluated by the aggregate precision of the effect estimates. 

For example, consider a clinical trial that compares three new drugs against a standard treatment. The goal is not simply to identify which drug performs best, but to estimate each drug’s causal effect relative to the control with high precision.  A drug that is superior, but yields only a marginal clinical benefit, may not be worth pursuing.  Similarly, cost, side effects, or scalability considerations may dominate small differences in efficacy. In such settings, precise effect estimation -- not just best-arm identification -- is essential.

Though originally derived as an optimal strategy for stratified sampling, the Neyman allocation is a canonical efficiency-oriented design, minimizing the statistical risk of a simple difference-in-means estimator \citep{Neyman1934}. Optimal assignment probabilities are proportional to arm-specific potential outcome standard deviations \citep{Neyman1934, Blackwell1979, Atkinson1982}. Because these quantities are unknown at the outset of an adaptive experiment, various papers have proposed schemes to approximate these quantities and update treatment probabilities accordingly. \cite{blackwell2022batch} proposes approximating these quantities with a pilot batch, while \cite{cai2024performance} notes that this approach can yield higher variance than a non-adaptive approach when the pilot is small. Other methods consider multistage \citep{zhao2023adaptive} and online \citep{dai2023clip, chen2025sigmoid} variants. 

When covariates are observed prior to treatment assignment, adaptive designs can be further enhanced to improve efficiency. A standard approach is to estimate arm-specific variances conditional on covariate values or strata. \cite{hahn2011adaptive} proposes implementing a Neyman allocation conditional on covariates by estimating outcome variances within treatment and control groups. \cite{tabord2018stratification} studies a two-stage randomized controlled trial in which strata are learned rather than prespecified. He proposes using pilot data to estimate a stratification scheme via tree-based methods, and selecting stratum-specific allocations using a Neyman-like approach. \cite{LiOwen2024} characterize optimal allocation rules that minimizes the asymptotic variance of a treatment effect estimator among covariate-adaptive designs, providing theoretical guarantees for efficiency gains relative to classical randomization. In the presence of covariates, alternative objectives may be considered. For example, \cite{Arbour2022} introduces a method based on weighted online discrepancy minimization in multiarm trials. The goal is to achieve covariate balance across arms, rather than explicitly to improve estimation precision. 

Another key design choice is whether the experiment is run for a fixed sample size or governed by a stopping rule. Adaptive trials may terminate early if one arm is shown to be sufficiently superior, or if continuing the experiment is deemed unlikely to change the decision. While early stopping can substantially reduce cost, it complicates statistical inference, as conventional estimators and confidence intervals may no longer be valid \citep{Johari2017}. A growing literature studies both optimal stopping rules and valid post-selection inference in adaptive experiments \citep{Ham2023}.

In this manuscript, the final causal estimates will be obtained by applying a Stein-like shrinkage estimator to the vector of $K$ treatment effect estimates. Our goal in this manuscript is to develop an adaptive allocation rule tailored to this class of estimator. In this way, our allocation method differs from other efficiency-oriented adaptive designs, which seek to minimize the variance of unbiased difference-in-means or Horvitz-Thompson estimators. This shift changes the design problem: assignment decisions should be evaluated by their impact on the aggregate risk of the shrinkage estimator, not by the marginal variance of each arm in isolation. 



\section{SURE and Candidate Estimators}\label{sec:theory}
\subsection{Notation and Assumptions}

Suppose the trial has a fixed, total sample size of $N$ individuals, indexed by $i = 1, \dots, N$. 
Denote by $W_i \in \{0, 1, \dots, K\}$ the treatment assigned to individual $i$, and define
\[ n_k = \sum_{i = 1}^N \I(W_i = k), \hspace{5mm} k = 0, \dots, K \] 
as the number of units assigned to treatment $k$. 

As discussed in Section~\ref{sec:intro}, we operate in the potential outcomes framework \citep{rubin1974estimating}, and assume units are sampled i.i.d. from a super-population. We also make the standard Stable Unit Treatment Value Assumption \citep[SUTVA;][]{rubin1980randomization}, meaning there is no interference across units and no hidden versions of the treatments. Under SUTVA, we observe outcomes 
\[
Y_i = Y_i(W_i), \hspace{5mm} i = 1, \dots, N,
\]
for each individual $i$.

We also need some notation for means and variances of potential outcomes. We define
\[ \mu_k = \E \left( Y(k) \right) \hspace{5mm} \text{ and } \hspace{5mm} V_k = \var \left( Y(k) \right), \] 
for $k = 0, \dots, K$, where the expectation and variance here are again over the population from which the experimental units are sampled. Under our definitions, we can define the $k^{th}$ treatment effect as the contrast
\[ \tau_k = \mu_k - \mu_0, \hspace{5mm} k = 1, \dots, K.\]

\subsection{Difference-in-Means Estimator}

The simplest choice of estimator is $\hbstau = \left( \hat \tau_1, \dots, \hat \tau_K \right)$ where 
\[ \hat \tau_k = \frac{1}{n_k} \sum_{i = 1}^n Y_i \cdot \I(W_i = k) - \frac{1}{n_0} \sum_{i = 1}^n Y_i \cdot \I(W_i = 0) \,.\] 
The difference-in-means estimator is unbiased, with known variance, i.e. 
\[ \E (\hbstau)  = \bstau \hspace{5mm} \text{ and } \hspace{5mm}\sigma_k^2 \equiv \var(\hat \tau_k) = \frac{V_k}{n_k} + \frac{V_0}{n_0}\, k = 1, \dots, K. \] 
The entries of $\hbstau$ are not independent, since each depends on the estimate from the control arm. It is straightforward to show that 
\[ \cov(\hat \tau_j, \hat \tau_k) = \frac{V_0}{n_0}, \hspace{10mm} \text{ for } j \neq k.\]

After the data is collected, we can invoke a Central Limit Theorem on the entries of $\hbstau$ as long as the potential outcomes $Y(k)$ are bounded and we have obtained a reasonably large number of samples in each arm. Recent work \citep[see e.g.][]{hadad2021confidence} has highlighted that adaptivity can undermine asymptotic normality if treatment selection is too responsive to estimated responses under each treatment. We explicitly assume this is not the case, such that, approximately, 
\begin{equation}\label{eq:samplingModel}
\hbstau \sim \mathcal{N} \left( \bstau,  \boldsymbol \Sigma \right) \hspace{5mm} \text{ where } \hspace{5mm} \boldsymbol \Sigma = \left( \begin{array}{cccc} 
\frac{V_1}{n_1} + \frac{V_0}{n_0} & \frac{V_0}{n_0} & \hdots & \frac{V_0}{n_0} \\
\frac{V_0}{n_0} & \frac{V_2}{n_2} + \frac{V_0}{n_0} & \hdots & \frac{V_0}{n_0}
 \\ \vdots & \vdots & \ddots & \vdots \\
 \frac{V_0}{n_0} & \frac{V_0}{n_0} & \hdots & \frac{V_K}{n_K} + \frac{V_0}{n_0} \end{array} \right) . 
\end{equation}

\noindent For further discussion of this assumption, see Section \ref{sec:cis}.

\subsection{Stein's Unbiased Risk Estimate}

In parallel estimation problems like this one, researchers can typically achieve reductions in statistical risk by using 
shrinkage estimation. 
The key idea is to introduce a small amount of bias by ``shrinking" the estimates toward a fixed point. 
The compensatory reduction in variance due to the regularization is such that one achieves a lower statistical risk than with an unbiased estimator. The most classic estimator of this type is the James-Stein estimator \citep{stein1956inadmissibility}, which shrinks an estimate of a multivariate normal mean vector towards zero. Though we assume approximate normality of $\hbstau$, the classical James-Stein estimator is not applicable to our setting, because its construction assumes the data are homoscedastic and that the entries are independent. 

There are a number of different adaptations of the James-Stein estimator to the settings of heteroscedasticity and dependence  \citep[see e.g.][]{xie2012sure}. Before considering the form of our adapted estimator, we introduce a useful tool for construction and evaluation of shrinkage estimators: Stein's Unbiased Risk Estimate \citep{stein1981estimation}. 

\begin{lemma}[SURE]\label{lemma:sure}
Suppose $\bsX \sim \mathcal{N}\left(\bsmu, \boldsymbol \Sigma\right) \in \mathbb{R}^K$ for non-degenerate $\bsSigma \in \mathbb{R}^{K \times K}$. We consider an estimator of $\bsmu \in \mathbb{R}^K$ of the form 
\[ \bsdelt(\bsX) = \bsX - g(\bsX), \]
for $g(\cdot)$ differentiable and $L_2$ integrable. 

Define $\mathcal{J}_{g}( \bsX )$ as the Jacobian matrix of $g(\cdot)$ evaluated at $\bsX$. Then, for loss function $\mathcal L(\bsmu, \boldsymbol v) =   || \bsmu - \boldsymbol v ||_2^2 $, we have 
\begin{equation}\label{eq:risk}
\e \left( \mathcal{L}\left( \bsdelt(\bsX), \bsmu \right) \right) = \Tr(\bsSigma ) + \e \bigg(|| g(\bsX)||_2^2 - 2 \cdot \Tr \big( \bsSig \mathcal{J}_{g}( \bsX )\big)  \bigg).
\end{equation}
Denote the corresponding expected squared-error risk by
\[ \mathcal{R}(\bsdelt(\bsX),\bsmu) = \e \left(\mathcal{L}(\bsdelt(\bsX),\bsmu)\right). \]
Hence, the estimator 
\begin{equation} \label{eq:SURE}
\text{SURE}(\bsdelt(\bsX)) = \Tr(\bsSigma ) + || g(\bsX)||_2^2 - 2 \cdot \Tr \big( \bsSig \mathcal{J}_{g}( \bsX )\big)
\end{equation} 
is unbiased for the risk $\mathcal{R} \left(\bsdelt(\bsX), \bsmu \right)$. 
\end{lemma}

\begin{proof*}
See the Appendix, Section \ref{sec:proofSureLemma}. 
\end{proof*}

\subsection{Candidate Estimators}

We consider several candidate shrinkage estimators to apply to the difference-in-means estimator $\hbstau$, all of which will work with our adaptive treatment allocation method. 

\textbf{Bock's estimator}. The shrinker of \cite{bock1975minimax} is explicitly designed to work with Gaussian vectors with dependent entries. 
Denoting as $\lambdamax$ the largest eigenvalue of $\bsSig$, the estimator is given by 
\[ \bsdeltb = \left(1 - \frac{\tilde p - 2}{\hbstau^\tran \bsSig^{-1} \hbstau} \right) \hbstau \hspace{5mm} \text{ where } \hspace{5mm} \tilde p = \frac{\Tr(\bsSig)}{\lambdamax}. \]
Note that $0 \leq \tilde p \leq K$. Bock showed that the estimator dominates $\hbstau$ if $\Tr(\bsSig) > 2\lambdamax$. The estimator is attractive for its simplicity, and admits several useful approximations. 

\textbf{SURE-Min estimator}. An alternative estimator can be obtained by 
choosing a shrinkage factor to directly minimize Stein's Unbiased Risk Estimate (Equation \ref{eq:SURE}), a tradition with a long history in the statistics literature \citep[see e.g.][]{li1985stein, xie2012sure}. We consider the class of  estimators that shrink $\hbstau$ toward zero by a constant factor $\rho \in \mathbb{R}$, i.e. estimators of the form
\[ \bsdelt(\hbstau) = (1 - \rho)\hbstau.\]
An unbiased estimate of the risk of such an estimator is 
\[ \text{SURE}(\bsdelt(\hbstau)) = \Tr(\bsSigma ) + \rho^2 || \hbstau ||_2^2 - 2\rho \cdot \Tr(\bsSig), \]
which has minimizer $\rho = \Tr(\Sigma) \big/ ||\hbstau||_2^2$. Thus, we define our second candidate estimator as 
\[ \bsdelts = \left( 1 - \frac{\Tr(\bsSig)}{||\hbstau||_2^2}\right) \hbstau, \]
where $|| \hbstau ||_2^2 = \sum_k \hat \tau_k^2$. We can also use Lemma \ref{lemma:sure} under which the risk of $\bsdelts$ is strictly less than that of $\hbstau$ regardless of the true value of $\bstau$. 


\begin{lemma}\label{lemma:riskBound_deltas}
$\bsdelts$ has risk strictly less than $\hbstau$ if 
\[ 4 \cdot \lambdamax <  \Tr(\bsSig).\]
\end{lemma}

\begin{proof*}
See the Appendix, Section \ref{sec:proof_deltas}.  
\end{proof*}
 
\textbf{Dimmery's estimator}. A final alternative is a modified version of the estimator proposed in \cite{dimmery2019shrinkage}. The authors consider 
online experiments in which the objective is to estimate the mean response under each treatment. They propose a heteroscedastic version of the James-Stein estimator that shrinks toward the grand mean. Because we consider contrasts between the mean response under each treatment and the mean response under control, we modify the estimator to shrink toward 0. We refer to this estimator as $\bsdeltad$, whose entries are given by 
\[ \hat \delta_{\text{D}, k} = \left(1 - \frac{(K - 2) \sigma_k^2}{|| \hbstau ||_2^2} \right) \hat \tau_k, \hspace{10mm} k = 1, \dots K. \] 
 
\noindent We can  again use Lemma \ref{lemma:sure} to establish a condition under which $\bsdeltad$ dominates $\hbstau$.


\begin{lemma}\label{lemma:riskBound_deltad}
$\btaujs$ has risk strictly less than $\hbstau$ if
\begin{equation}\label{eq:domCondition}
\frac{1}{2} \left( \max_k \sigma_k^2 \right) \left( (K - 2) (\max_k \sigma_k^2) + 4 \lambdamax  \right) \leq \sum_k \sigma_k^4.
\end{equation}
\end{lemma}

\begin{proof*}
See the Appendix, Section \ref{sec:proof_deltad}. 
\end{proof*}

\noindent As with the maximum eigenvalue condition for Bock's estimator, Lemmas \ref{lemma:riskBound_deltas} and  \ref{lemma:riskBound_deltad} are useful for practitioners in that they provide a condition that can be checked at the end of the online trial to obtain a guarantee that the risk will be lower than that of $\hbstau$. 

In practice, we can truncate the shrinkage factor for each of these estimators so that it cannot go below $0$. These ``positive-part" analogues are known to have lower risk than their ``smooth" counterparts \citep{baranchik1964multiple}. However, positive-part estimators no longer admit risk expressions as expectations of ratios of Gaussian quadratic forms, because the truncation induces expectations over truncated Gaussian regions. For this reason, our theoretical analysis focuses instead on the original smooth estimators, whose risks may be seen as an upper bound on the positive-part versions. 

\section{Oracle Designs}\label{sec:oracle}

Our goal is to design a treatment allocation approach that minimizes the risk of each candidate estimator.  We begin by analyzing the non-adaptive setting with $N$ units, and consider ``oracle" designs in which the potential outcome means $\mu_0, \dots, \mu_K$ and variances $V_0, \dots, V_K$ are known. In practice, these parameters must be estimated from the data, but the oracle designs give us intuition about the optimal risk-minimizing allocations.  

We focus on the regime in which the signal-to-noise ratio (SNR) is low. 
This assumption restricts to the regime where our shrinkage estimators can yield meaningful performance gains, as is widely discussed in the Empirical Bayes literature \citep[see e.g.][]{efron2012large}. In less noisy regimes, applying a shrinkage estimator provides less benefit.


\subsection{A Modified Neyman Allocation}

In the standard Neyman allocation, the treatment probability is proportional to each standard deviation $\sqrt{V_k}$. This approach minimizes the expected squared error when estimating the mean for each arm. However, our goal is slightly different: we want to minimize the risk when estimating each treatment effect, which is a contrast between the active treatment arms and the control arm. If we use the difference-in-means estimator $\hbstau = \left( \hat \tau_1, \dots, \hat \tau_K \right)$, the risk is


\[ \mathcal{R}\left( \hbstau, \bstau \right) = \Tr(\bsSig) = \sum_{k = 1}^K \frac{V_0}{n_0} + \frac{V_k}{n_k} = \frac{K V_0}{n_0} + \sum_{k = 1}^K \frac{V_k}{n_k}.  \]
Hence, up to rounding, the risk-minimizing allocation of units to treatment arms is the solution to the following optimization problem (Problem \ref{eq:opt-prob}): 

\begin{equation}\label{eq:opt-prob}
\begin{aligned}
\min_{n_0,\dots,n_K}\quad & \frac{K V_0}{n_0} + \sum_{k=1}^K \frac{V_k}{n_k} \\
\text{subject to}\quad & \sum_{k = 0}^K n_k = N, \\
& 0 < n_0, \dots, n_K. 
\end{aligned}
\end{equation}
This is a convex problem whose solution is a slight modification to the standard Neyman Allocation. 

\begin{lemma}\label{lemma:neymanAlloc}
When using the difference-in-means estimator, the risk-minimizing treatment allocation is 
\[ n_0^\star = \frac{N \sqrt{K V_0}}{\sqrt{K V_0} + \sum_{\ell = 1}^K \sqrt{V_\ell}} \hspace{10mm} \text{ and } \hspace{10mm} n_k^\star = \frac{N \sqrt{V_k}}{\sqrt{K V_0} + \sum_{\ell = 1}^K \sqrt{V_\ell}} \text{ for $k = 1, \dots, K$.}  \]
\end{lemma}

\begin{proof*}
See the Appendix, Section \ref{sec:proof_neymanAlloc}.
\end{proof*}

\subsection{Shrinker Allocations in the Low Signal-to-Noise Ratio Regime}

\subsubsection{Approximate Minimization}

The Neyman Allocation provides a useful starting place for choosing a treatment allocation that minimizes the risk using a shrinkage estimator. For any $g(\bsX)$ satisfying the conditions of Lemma \ref{lemma:sure}, the allocation which minimizes the risk of estimator  
\[ \bsX - g(\bsX), \] 
is encoded in the following optimization problem (Problem \ref{eq:opt-prob-enhanced}):

\begin{equation}\label{eq:opt-prob-enhanced}
\begin{aligned}
\min_{n_0,\dots,n_K}\quad & \frac{K V_0}{n_0} + \sum_{k=1}^K \frac{V_k}{n_k} + \e \bigg(|| g(\bsX)||_2^2 - 2 \cdot \Tr \big( \bsSig \mathcal{J}_{g}( \bsX )\big)\bigg)\\
\text{subject to}\quad & \sum_{k = 0}^K n_k = N, \\
& 0 < n_0, \dots, n_K. 
\end{aligned}
\end{equation}
In this form, we can interpret 
\[ \Delta = \e \bigg(|| g(\bsX)||_2^2 - 2 \cdot \Tr \big( \bsSig \mathcal{J}_{g}( \bsX )\big)\bigg)\]
as a regularization term added to the objective in Problem \ref{eq:opt-prob}. 

Assuming the regularization effect is small, the solutions to Problem \ref{eq:opt-prob-enhanced} will be ``local" to the Neyman allocation $\boldsymbol{n^\star} = (n_0^\star, \dots, n_K^{\star})$. 
Hence, we can characterize the optimizers of Problem \ref{eq:opt-prob-enhanced} by analyzing the behavior of the objective near the Neyman allocation. Denote as $\boldsymbol{n^\dag} = (n_0^\dag, \dots, n_K^{\dag})$ as a generic solution to Problem \ref{eq:opt-prob-enhanced}. A simple approximation yields 

\begin{equation}\label{eq:solApprox}
\frac{n_k^\dag - n_k^\star}{n_k^\star} \approx \frac{1}{2 \sna^2} \left[ \frac{1}{N} \sum_{j=0}^K n_j^\star \left. \frac{\partial \Delta}{\partial n_j} \right|_{\boldsymbol n = \boldsymbol n^\star} - \left. \frac{\partial \Delta}{\partial n_k} \right|_{\boldsymbol n = \boldsymbol n^\star} \right] \hspace{5mm} \text{ where } \hspace{5mm}\sna = \frac{1}{N} \left( \sqrt{K V_0} + \sum_{k = 1}^K \sqrt{V_k} \right). 
\end{equation}

For a derivation, see the Appendix, Section~\ref{appendix:solApprox}. This result encodes an intuitive result: relative to the Neyman allocation, minimizers of the shrinker risk will allocate more units to arms that more rapidly decrease the value of the regularization term. We use this result to characterize minimizers of shrinker risks. 

\subsubsection{Bock's Estimator}

We first consider Bock's estimator, $\bsdeltb$. The regularization term for $\bsdeltb$ admits a simple representation in the low signal-to-noise ratio regime, allowing us to obtain the following result. 

\begin{lemma}\label{lemma:bockRiskZero}
Suppose $K>2$ and
\begin{equation}\label{eq:riskConditionBock}
\left( \max_{k\geq 1}\sqrt{V_k}-\min_{k\geq 1}\sqrt{V_k} \right) \leq\frac{1}{2}\sqrt{KV_0}. 
\end{equation}
Then there exists $\epsilon>0$ such that, whenever
\[ \kappa =\bstau^\tran\bsSigmaN^{-1}\bstau<\epsilon, \]
the derivatives of the regularization term for Bock's estimator, evaluated at the modified Neyman allocation, satisfy
\[ \left.\frac{\partial\Delta}{\partial n_0}\right|_{\bsn=\bsn^\star}<\min_{k\geq1}\left.\frac{\partial\Delta}{\partial n_k}\right|_{\bsn=\bsn^\star}. \]
Moreover, among the active treatment arms, if $V_k < V_j$ for some $1 \leq j, k \leq K$, then 
\[ \left.\frac{\partial\Delta}{\partial n_j}\right|_{\bsn=\bsn^\star}<\left.\frac{\partial\Delta}{\partial n_k}\right|_{\bsn=\bsn^\star}. \]
\end{lemma}

\begin{proof*}
See the Appendix, Section~\ref{sec:proofBockRisk}.
\end{proof*}

Lemma \ref{lemma:bockRiskZero} tells us that under a mild heteroscedasticity constraint and in the low SNR regime, the regularization term gradient will be more negative for the control arm than any of the active treatment arms. Moreover, among the active treatment arms, the regularization term gradient is smaller for arms with larger potential outcome variances. Combined with Approximation~\eqref{eq:solApprox}, this suggests that the risk-minimizing allocation for Bock's estimator will allocate more units to the control arm than the Neyman Allocation while shifting proportionally fewer units away from higher-variance treatment arms. This is confirmed empirically in Section \ref{sec:nonadaptSim}.

\subsubsection{SURE-Min and Dimmery Estimators}

The SURE-min and Dimmery estimators, $\bsdelts$ and $\bsdeltad$, do not admit exact closed-form risk expressions at $\bstau=\boldsymbol 0$. For $\bsdelts$, we can conduct a heuristic analysis at a low signal-to-noise ratio and use Approximation~\eqref{eq:solApprox} to approximately characterize the risk-minimizing allocations. The more complex form of $\bsdeltad$ does not allow for risk approximations that are analytically tractable; hence, we defer analysis of its risk-minimizing allocation to the empirical analysis in the next section. 

We apply a series of moment-matching approximations to the risk of $\bsdelts$, and invoke our low-SNR assumption, to obtain the approximate regularization term  

\begin{equation}\label{eq:steinminRiskLowSNR}
\Delta \approx \Tr(\bsSig) \left( \frac{4\Tr(\bsSig^2)-\Tr(\bsSig)^2}
{\Tr(\bsSig)^2-2\Tr(\bsSig^2)} \right).
\end{equation}
Using the Chain Rule, we differentiate Expression~\eqref{eq:steinminRiskLowSNR} with respect to $n_j$, yielding
\[ \frac{\partial\Delta}{\partial n_j} \approx \frac{\partial\Delta}{\partial\Tr(\bsSig)} \frac{\partial\Tr(\bsSig)}{\partial n_j} + \frac{\partial\Delta}{\partial\Tr(\bsSig^2)} \frac{\partial\Tr(\bsSig^2)}{\partial n_j}. \]
At the Neyman allocation, $\left. \partial\Tr(\bsSig)/\partial n_k\right|_{\bsn=\bsn^\star}
=-\sna^2$
for every $k=0,\dots,K$. Thus, after substituting the gradient into
Approximation~\eqref{eq:solApprox}, all terms involving
$\partial\Delta/\partial\Tr(\bsSig)$ cancel, and we obtain
\[ \frac{n_k^\dag-n_k^\star}{n_k^\star} \approx \frac{1}{2\sna^2} \left. \frac{\partial\Delta}{\partial\Tr(\bsSig^2)} \right|_{\bsn=\bsn^\star} \left[ \frac{1}{N}\sum_{j=0}^K n_j^\star \left. \frac{\partial\Tr(\bsSig^2)}{\partial n_j} \right|_{\bsn=\bsn^\star} - \left. \frac{\partial\Tr(\bsSig^2)}{\partial n_k} \right|_{\bsn=\bsns} \right]. \]

Note that the factor,  $\left.\partial\Delta/\partial\Tr(\bsSig^2)\right|_{\bsn=\bsns}$, must be positive. Next, we observe
\[ \Tr(\bsSig^2) = K^2 \frac{V_0^2}{n_0^2} +2\frac{V_0}{n_0} \sum_{k = 1}^K \frac{V_k}{n_k} + \sum_{k = 1}^K \frac{V_k^2}{n_k^2}. \]
Differentiating and evaluating at the Neyman allocation, we obtain
\begin{equation}\label{eq:trSigSqDerivs}
\left. \frac{\partial \Tr(\bsSig^2)}{\partial n_0}\right|_{\bsns} = -2 \sna^3 \left( \sqrt{K V_0} + \frac{1}{K}\sum_{k = 1}^K \sqrt{V_k} \right) \hspace{3mm} \text{ and } \hspace{3mm}  \left. \frac{\partial \Tr(\bsSig^2)}{\partial n_k}\right|_{\bsns} = -2 \sna^3 \left(\sqrt{\frac{V_0}{K}} + \sqrt{V_k} \right).
\end{equation}

The expressions in \eqref{eq:trSigSqDerivs} allow us to assess which arms will be assigned more units. Unless $V_0$ is extremely small relative to the other $V_k$ values, the derivative with respect to $n_0$ is more negative than the derivatives with respect to the treatment-arm sample sizes. Hence, the approximate $\mathcal{R}(\bsdelts, \bstau)$-minimizing allocation will again shift units toward the control arm relative to the Neyman allocation. For additional technical details, see the Appendix, Section \ref{sec:riskApprox}.

\subsection{Simulated Static Allocations}\label{sec:nonadaptSim}

To complement our heuristic analysis, we also directly compute the risk-minimizing allocations across several different simulation regimes. In each simulated setting, we also compute the non-adaptive Neyman allocation $\boldsymbol n^\star$, and assess how minimizing the risks of $\bsdeltb, \bsdelts,$ and $\bsdeltad$ perturbs this baseline allocation.


\subsubsection{Computing Risk-Minimizing Allocations}

For notational convenience, we write $\mathcal R(\bsdelt, \bstau ; \bsn,\boldsymbol V)$ to represent risk of estimator $\bsdelt$ evaluated at allocation $\bsn$ under potential outcome variances $\boldsymbol{V}$ and true treatment effects $\bstau$. Notably, $\mathcal R(\bsdelt, \bstau ; \bsn,\boldsymbol V)$
is \emph{not} a convex function of the allocations $n_0, \dots, n_K$ for any of our shrinkage estimators. Hence, we search for the risk-minimizing allocations through a greedy swapping algorithm. Define
\[ \boldsymbol{n^{(j)}} = \{ n_{0}^{(j)}, \dots, n_{K}^{(j)}\} \in \mathbb{Z}^{K +1} \] 
as the allocation of units to treatment arms at iteration $j$ of the algorithm. Next, define the set of potential swaps
\[ S^{(j)} = \{ \boldsymbol n \in \mathbb{Z}^{K + 1} \mid \text{ $\boldsymbol n$ swaps exactly one unit across treatment arms from } \boldsymbol{n^{(j)}} \}\,. \] 
Because there are $K + 1$ treatment arms, the ``swap set" $S^{(j)} $ will contain $K(K+1)$ possible allocations. At each iteration of the algorithm, we simply move in the direction of the swap that most reduces the risk. This approach is encoded below, in Algorithm~\ref{greedyAlgorithm}:

\begin{equation}\label{greedyAlgorithm}
\begin{aligned}
& \texttt{Start with  the Neyman allocation $\boldsymbol{n^{(0)}} = \boldsymbol{n^\star}$.} \hspace{55mm} \\
& \texttt{For iteration $j = 1, 2, \dots$:}\\
& \hspace{5mm} \texttt{For each potential allocation $\bsn \in S^{(j-1)}$:}\\
& \hspace{10mm} \texttt{Compute $\mathcal R(\bsdelt, \bstau ; \bsn,\boldsymbol V)$. } \\
& \hspace{5mm} \texttt{Set $\boldsymbol{n^{(j)}} = \argmin_{\bsn \in S^{(j - 1)}} \mathcal R(\bsdelt, \bstau ; \bsn,\boldsymbol V)$} \\
& \hspace{5mm} \texttt{If $\mathcal R(\bsdelt, \bstau ; \boldsymbol{n^{(j)}},\boldsymbol V) \geq \mathcal R(\bsdelt, \bstau ; \boldsymbol{n^{(j-1)}},\boldsymbol V)$}\\
& \hspace{10mm} \texttt{Return $\boldsymbol{n^{(j-1)}}$.}
\end{aligned}
\end{equation}
For details on how $\mathcal R(\bsdelt, \bstau ; \boldsymbol{n^{(j)}},\boldsymbol V)$ is computed, see Section~\ref{sec:effRisk}. Although Algorithm~\ref{greedyAlgorithm} could plausibly become trapped in a local minimum, we find that for the sample sizes considered in this section, its output is unchanged across a variety of starting points.

\subsubsection{Simulation Setup}\label{sec:simSetup}

We simulate data given a variety of data-generating regimes. Throughout, we fix the total sample size at $N=1,000$ and compute, for each setting, the Neyman allocation $\boldsymbol n^\star$ as well as the allocations that minimize the risks of $\bsdeltb$, $\bsdelts$, and $\bsdeltad$ using the greedy swapping approach (Algorithm~\ref{greedyAlgorithm}).

We alter four key features across our simulations: the number of treatment arms $K$; the magnitude of the control potential outcome variance $V_0$ relative to the active treatment-arm variances $\{V_k\}_{k=1}^K$; the sparsity of $\bstau$; and the signal-to-noise ratio at the Neyman allocation, $\kappa=\bstau^\top\bsSigmaN^{-1}\bstau$. 

To generate the data, we first draw the potential outcome variances $V_k$ for the active treatment arms as independent and identically distributed samples from a log normal distribution with log-location parameter of $\log(350)$ and log-standard-deviation parameter $0.60$. These values are then sorted so that $V_1 = V_{(1)}, \dots, V_k = V_{(K)}$; this sorting will later help us to reason about reallocations across active treatment arms according to their relative variances. Next, we consider two control-variance regimes.  In the low control-variance regime, $V_0$ is set at half of the average potential outcome variance of the treatment arm, i.e. $V_0 = \tfrac{1}{2}\cdot K^{-1}\sum_{k=1}^K V_k$. In the high control-variance regime, $V_0$ is equal to four times the average variance of the treatment arm, i.e. $V_0 = 4\cdot K^{-1}\sum_{k=1}^K V_k$. 

The treatment effects $\bstau$ are next sampled to ensure the desired sparsity and signal-to-noise-ratio. We consider two configurations for the shape of the ``signal'' $\bstau$. In the dense setting, an unscaled $\tilde{\bstau}$ is sampled i.i.d. from a $\text{Unif}([1, 2])$ distribution, resulting in all causal effects having comparable magnitude. In the sparse setting, $\tilde{\bstau} = (1, \dots, 0)$ has only one nonzero entry, which corresponds to the first arm (and thus, due to the sorting, also corresponds to the smallest potential outcome variance among the active treatment arms). 

Once the unscaled effects are sampled, they are scaled to produce a $\bstau$ at the desired signal-to-noise-ratio $\kappa = \bstau^\tran \bsSigmaN^{-1} \bstau$. $\bsSigmaN$ is the covariance matrix at the Neyman allocation, so $\kappa$ is an approximation to the signal-to-noise ratio at the true minimizer, as the covariance matrix itself changes as we reallocate units across treatment arms. 

Tables \ref{tab:simAll_control} and \ref{tab:simAll_treat1} report the results of all the simulations, across the grid of regimes:
\[
K\in\{6,12\},\qquad
V_0\in\{\texttt{low},\texttt{high}\},\qquad 
\bstau\in\{\texttt{dense},\texttt{sparse}\}, \qquad \kappa\in\{0,2,4\}.
\]
Note that when $\kappa = 0$, we must have $\bstau = \boldsymbol 0$, so $\bstau$ is neither dense nor sparse. This setting is reflected in the tables with a ``$-$" symbol in the ``$\bstau$ shape" column. 

\subsubsection{Results}\label{sec:simResults}

In Tables \ref{tab:simAll_control} and \ref{tab:simAll_treat1}, we show how the risk-minimizing allocations differ from the Neyman allocation. Several trends are evident. 


\paragraph{Large, consistent reallocations to the control arm.}
Table~\ref{tab:simAll_control} gives the share of units assigned to the control arm under the Neyman allocation, and shows the increase in the control share for each of the risk-minimizing allocations. The clearest trend is that the shrinker risk-minimizing designs always allocate more units to control than the Neyman allocation. This pattern holds across all dimensions, $V_0$ values, sparsity settings, and signal-to-noise ratios. 

In nearly all cases, the design minimizing the risk of Bock's estimator reassigns the most units to the control arm, while reassignment volumes are smaller for designs minimizing the risk of the SURE-min and Dimmery estimators. The only exceptions occur when $V_0$ is high, $\bstau$ is dense, and $\kappa>0$, where the SURE-min design reallocates the same or slightly more units to control than the Bock design.

The magnitude of these reallocations is also mediated by signal strength. Especially for Bock's estimator, they are typically largest when $\kappa=0$, tend to attenuate at a moderate signal-to-noise ratio ($\kappa = 2$), and become smallest when the signal-to-noise ratio is large $(\kappa = 4)$. Reallocations toward control are also typically larger when there are more active treatment arms ($K = 12$ versus $K = 6$).

The control-variance regime affects the baseline Neyman allocation to control, but has surprisingly little effect on reallocations toward control for Bock's estimator. However, the low-control-variance regime tends to yield smaller control reallocations for the $\bsdelts$- and $\bsdeltad$-risk minimizing designs. 

Lastly, the magnitude of the reallocation is mediated by the sparsity of the treatment effects. When $\kappa>0$, reallocations toward control under $\bsdeltb$ are substantially larger when $\bstau$ is sparse than when it is dense, often by a wide margin. In several settings, sparse signals lead to control reallocations for Bock that are comparable to those observed at $\kappa=0$, whereas dense signals produce more moderate shifts. The same interaction is present, though markedly less dramatic, for the SURE-min and Dimmery estimators. 

Broadly, this pattern mirrors our heuristic investigation, which showed that risk-minimizing allocations for all estimators favor more control units relative to the Neyman allocation $\boldsymbol{n^\star}$. The simulations also reflect a somewhat intuitive pattern. Increasing the allocation to control arms will tend to de-correlate the entries of $\hbstau$ by reducing the magnitude of the off-diagonal entries in $\bsSig$, which are all equal to $V_0/n_0$. This effectively increases the amount of independent information available to the shrinkage estimators.

\paragraph{Modest reallocations across active treatment arms.}
The risk-minimizing designs also reallocate units across the active treatment arms, though these changes are modest in magnitude relative to reallocations toward control. Table \ref{tab:simAll_treat1} reports reallocations toward treatment arm 1 and toward treatment arm $K$, measured relative to the Neyman allocation. Recall that treatment arm 1 has the smallest potential outcome variance among active treatments and the only nonzero signal when $\bstau$ is sparse. Treatment arm $K$ has the largest potential outcome variance among active treatments. 

The minimizing allocations for Bock's estimator always decrease the share of units assigned to the first active treatment arm and almost always decrease the share assigned to the $K^{th}$ treatment arm. The Neyman allocation to the first active treatment arm can be quite small, especially when $K = 12$, so these reductions can leave a very small proportion of units allocated to the first treatment arm. The behavior is largely insensitive to the dimension, control-arm variance, and signal-to-noise ratio. Reallocations away from the treatments arms are larger when the signal is sparse than when it is dense, corresponding to the increased preference for control units in this setting. 

The minimizing allocations for the SURE-min estimator are somewhat more variable. The risk minimizer allocates units away from treatment arm 1 when the signal-to-noise ratio is small, but tends to allocate units toward the arm when $\kappa$ grows, especially when the signal is sparse. Meanwhile, units are almost never allocated toward the $K^{th}$ treatment arm, and are typically allocated away from it. 

Dimmery's estimator exhibits the most variable behavior across simulation conditions. When the signal-to-noise ratio is zero or when the signal is dense, the minimizer either allocates fewer units to the first treatment arm, or just a few more units. However, in the case of a sparse signal (which is concentrated in the first treatment arm), Dimmery's estimator allocates many more units to this arm. Hence, Dimmery's estimator exhibits a strong preference for ``overweighting'' the arm with a strong signal, even though that arm exhibits low potential outcome variances. $\bsdeltad$-risk-minimizing allocations tend to exhibit small changes (both positive and negative) in allocations to the $K^{th}$ treatment arm.

\begin{table}[ht]
\centering
\begin{threeparttable}
\caption{Risk-minimizing allocations for $N = 1{,}000$: control share.}
\label{tab:simAll_control}

\small
\setlength{\tabcolsep}{8pt}

\begin{tabular}{cccccrrr}
\toprule
& & & & \multicolumn{1}{c}{Neyman}
& \multicolumn{3}{c}{$\Delta n_0/N$ (\textit{pp})} \\
\cmidrule(lr){6-8}
$K$ & $V_0$ regime & $\bstau$ shape & $\kappa$
& $n_0^\star/N$ (\%)
& Bock & SURE-min & Dimmery \\
\midrule
6 & high & $-$ & 0 & 38\% & +23 & +13 & +10 \\
6 & high & dense & 2 & 38\% & +11 & +11 & +7 \\
6 & high & sparse & 2 & 38\% & +19 & +10 & +8 \\
6 & high & dense & 4 & 38\% & +6 & +9 & +6 \\
6 & high & sparse & 4 & 38\% & +16 & +8 & +8 \\
\midrule
6 & low & $-$ & 0 & 20\% & +22 & +6 & +4 \\
6 & low & dense & 2 & 20\% & +12 & +6 & +4 \\
6 & low & sparse & 2 & 20\% & +19 & +4 & +4 \\
6 & low & dense & 4 & 20\% & +7 & +6 & +4 \\
6 & low & sparse & 4 & 20\% & +16 & +3 & +4 \\
\midrule
12 & high & $-$ & 0 & 29\% & +38 & +15 & +14 \\
12 & high & dense & 2 & 29\% & +16 & +12 & +10 \\
12 & high & sparse & 2 & 29\% & +34 & +14 & +11 \\
12 & high & dense & 4 & 29\% & +11 & +10 & +9 \\
12 & high & sparse & 4 & 29\% & +30 & +12 & +11 \\
\midrule
12 & low & $-$ & 0 & 15\% & +34 & +7 & +6 \\
12 & low & dense & 2 & 15\% & +22 & +6 & +5 \\
12 & low & sparse & 2 & 15\% & +32 & +6 & +5 \\
12 & low & dense & 4 & 15\% & +14 & +6 & +5 \\
12 & low & sparse & 4 & 15\% & +29 & +5 & +6 \\
\bottomrule
\end{tabular}

\footnotesize
$\Delta n_0$ represents the change in the proportion of units allocated to the
control arm relative to the Neyman allocation $n_0^\star$. The Neyman column reports
percentages, while the remaining columns report percentage-point changes (\textit{pp}).

\end{threeparttable}
\end{table}

\begin{table}[ht]
\centering
\begin{threeparttable}
\caption{Risk-minimizing allocations for $N = 1{,}000$: treatment-arm shares.}
\label{tab:simAll_treat1}

\small
\setlength{\tabcolsep}{4pt}

\begin{tabular}{ccccrrrrrrrr}
\toprule
& & & &
\multicolumn{1}{c}{Neyman}
& \multicolumn{3}{c}{$\Delta n_1/N$ (\textit{pp})}
& \multicolumn{1}{c}{Neyman}
& \multicolumn{3}{c}{$\Delta n_K/N$ (\textit{pp})} \\
\cmidrule(lr){6-8} \cmidrule(lr){10-12}
$K$ & $V_0$ regime & $\bstau$ shape & $\kappa$
& $n_1^\star/N$ (\%)
& Bock & SURE-min & Dimmery
& $n_K^\star/N$ (\%)
& Bock & SURE-min & Dimmery \\
\midrule
6 & high & $-$ & 0 & 5\% & $-3$ & $-2$ & $-3$ & 14\% & $-4$ & $-2$ & $-$ \\
6 & high & dense & 2 & 5\% & $-1$ & $-2$ & $-3$ & 14\% & $-2$ & $-1$ & $+1$ \\
6 & high & sparse & 2 & 5\% & $-3$ & $-$ & $+4$ & 14\% & $-3$ & $-2$ & $-2$ \\
6 & high & dense & 4 & 5\% & $-1$ & $-2$ & $-2$ & 14\% & $-1$ & $-2$ & $-3$ \\
6 & high & sparse & 4 & 5\% & $-3$ & $+2$ & $+6$ & 14\% & $-3$ & $-2$ & $-3$ \\
\midrule
6 & low & $-$ & 0 & 9\% & $-4$ & $-2$ & $-2$ & 18\% & $-2$ & $+1$ & $+2$ \\
6 & low & dense & 2 & 9\% & $-2$ & $-2$ & $-2$ & 18\% & $-1$ & $-$ & $-$ \\
6 & low & sparse & 2 & 9\% & $-4$ & $+2$ & $+9$ & 18\% & $-2$ & $-$ & $-2$ \\
6 & low & dense & 4 & 9\% & $-2$ & $-1$ & $+2$ & 18\% & $-$ & $-$ & $+2$ \\
6 & low & sparse & 4 & 9\% & $-4$ & $+4$ & $+12$ & 18\% & $-1$ & $-1$ & $-4$ \\
\midrule
12 & high & $-$ & 0 & 4\% & $-3$ & $-2$ & $-2$ & 8\% & $-4$ & $-1$ & $-$ \\
12 & high & dense & 2 & 4\% & $-1$ & $-1$ & $-$ & 8\% & $-2$ & $-1$ & $+1$ \\
12 & high & sparse & 2 & 4\% & $-3$ & $-1$ & $+5$ & 8\% & $-3$ & $-1$ & $-$ \\
12 & high & dense & 4 & 4\% & $-1$ & $-1$ & $+2$ & 8\% & $-1$ & $-1$ & $-1$ \\
12 & high & sparse & 4 & 4\% & $-3$ & $-$ & $+8$ & 8\% & $-3$ & $-1$ & $-1$ \\
\midrule
12 & low & $-$ & 0 & 4\% & $-3$ & $-1$ & $-2$ & 9\% & $-2$ & $-$ & $+1$ \\
12 & low & dense & 2 & 4\% & $-2$ & $-1$ & $-1$ & 9\% & $-1$ & $-$ & $+1$ \\
12 & low & sparse & 2 & 4\% & $-3$ & $-$ & $+7$ & 9\% & $-2$ & $-$ & $-$ \\
12 & low & dense & 4 & 4\% & $-1$ & $-$ & $+1$ & 9\% & $-1$ & $-$ & $-$ \\
12 & low & sparse & 4 & 4\% & $-3$ & $+1$ & $+10$ & 9\% & $-2$ & $-$ & $-1$ \\
\bottomrule
\end{tabular}

\footnotesize
$\Delta n_1$ and $\Delta n_K$ denote the change in the proportion of units allocated to treatment arms~1 and $K$, respectively, relative to the Neyman allocation values $n_1^\star$ and $n_K^\star$. The Neyman columns report percentages, while the remaining columns report percentage-point changes (\textit{pp}). Positive values indicate reallocations toward the corresponding arm. For readability, values that round to $0$ are replaced with $-$. 
\end{threeparttable}
\end{table}

Figures~\ref{fig:alloc_controlshift_delta0}, \ref{fig:alloc_controlshift_delta9}, and \ref{fig:alloc_armspecific_sparse} provide three representative visual summaries of risk-minimizing allocations when $K = 12$. Figure \ref{fig:alloc_controlshift_delta0} demonstrates the tendency of all shrinker-risk-minimizers -- especially Bock's estimator $\bsdeltb$ -- to reallocate sample toward the control arm. Figure \ref{fig:alloc_controlshift_delta9} shows how this pattern attenuates at larger values of the signal-to-noise ratio, especially for $\bsdelts$ and $\bsdeltad$. Figure \ref{fig:alloc_armspecific_sparse} reflects how, other things equal, the $\bsdeltad$-minimizing allocation tends to assign units to the arm with strong signal when $\bstau$ is sparse. 

\begin{figure}[ht]
\centering
\includegraphics[width=0.95\linewidth]{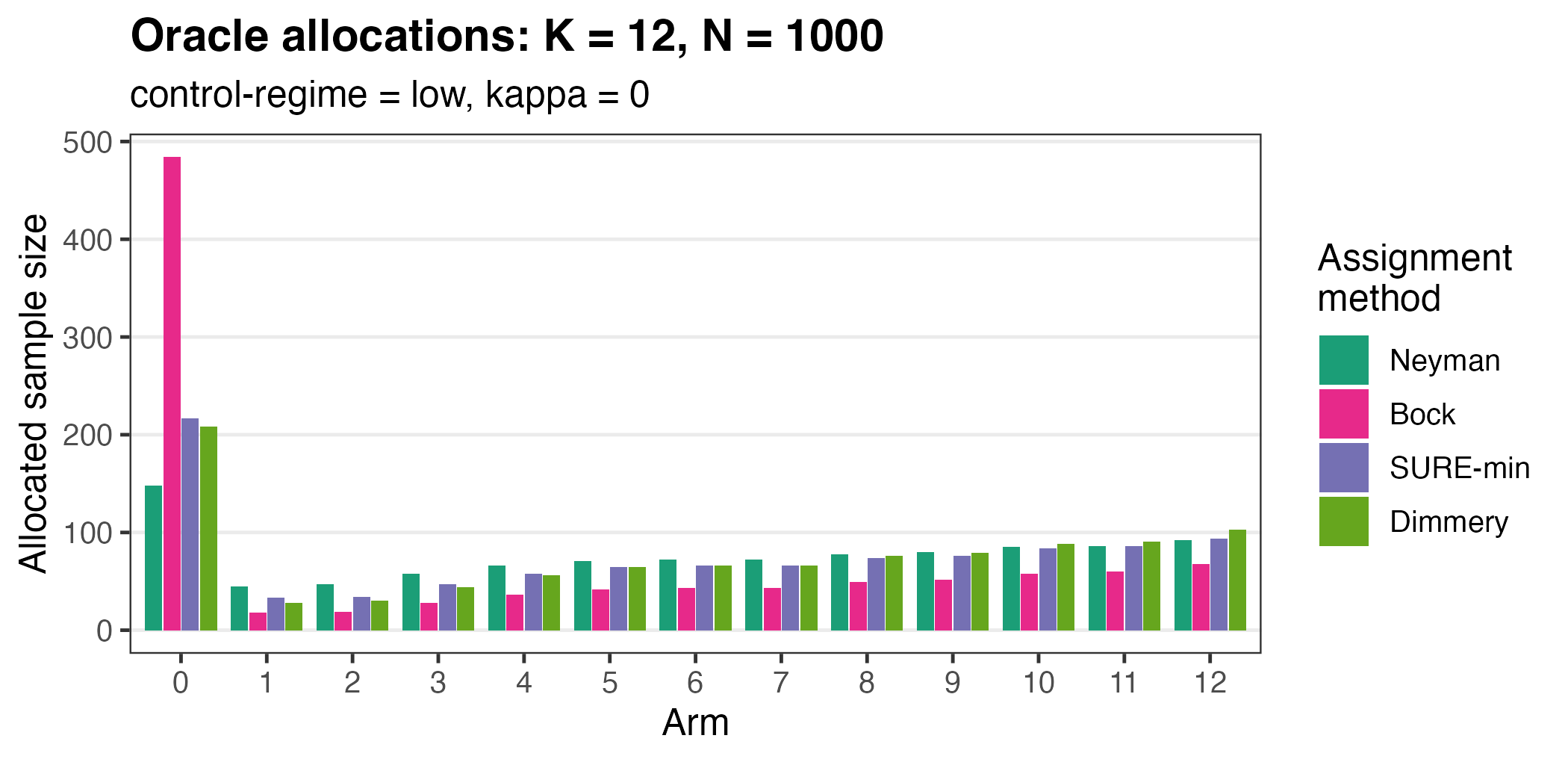}
\caption{Oracle allocations for $K=12$ and $N=1,000$ with low control variance, and $\kappa=0$. All shrinker-risk-minimizing designs shift units toward the control arm relative to the Neyman allocation, with the largest control shift under Bock's estimator $\bsdeltb$. SURE-min $\bsdelts$ and Dimmery $\bsdeltad$ exhibit the same trend but with smaller magnitude.}
\label{fig:alloc_controlshift_delta0}
\end{figure}

\begin{figure}[ht]
\centering
\includegraphics[width=0.95\linewidth]{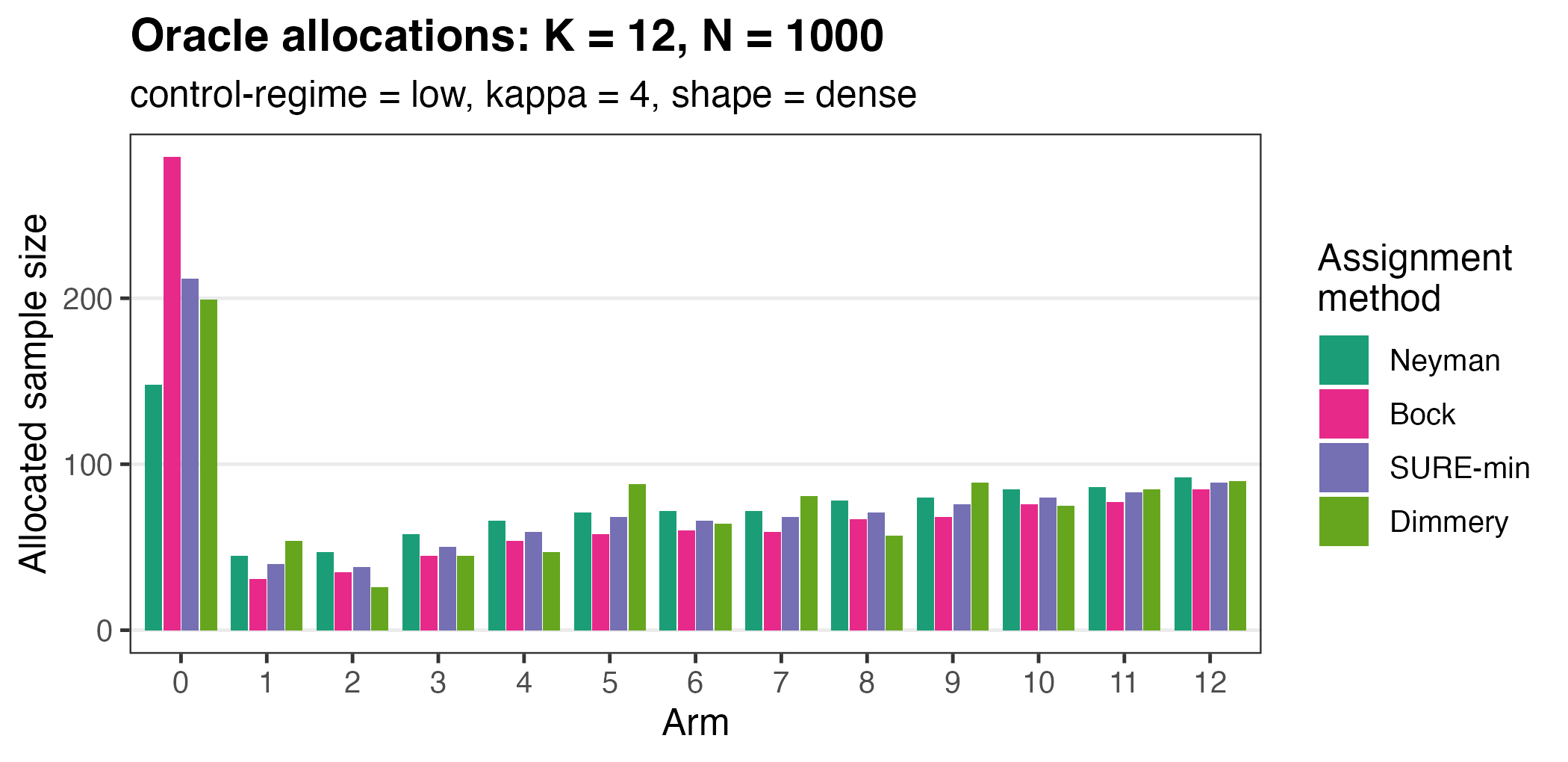}
\caption{Oracle allocations for $K=12$ and $N=1,000$ with low control variance, $\kappa=4$, and dense $\boldsymbol{\tau}$. Compared to the $\kappa=0$ case (Figure~\ref{fig:alloc_controlshift_delta0}), the shrinker preference for allocating additional units to control is substantially reduced for all estimators, especially the Bock estimator (though it still favors the control arm most among the three shrinkers).}
\label{fig:alloc_controlshift_delta9}
\end{figure}

\begin{figure}[ht]
\centering
\includegraphics[width=0.95\linewidth]{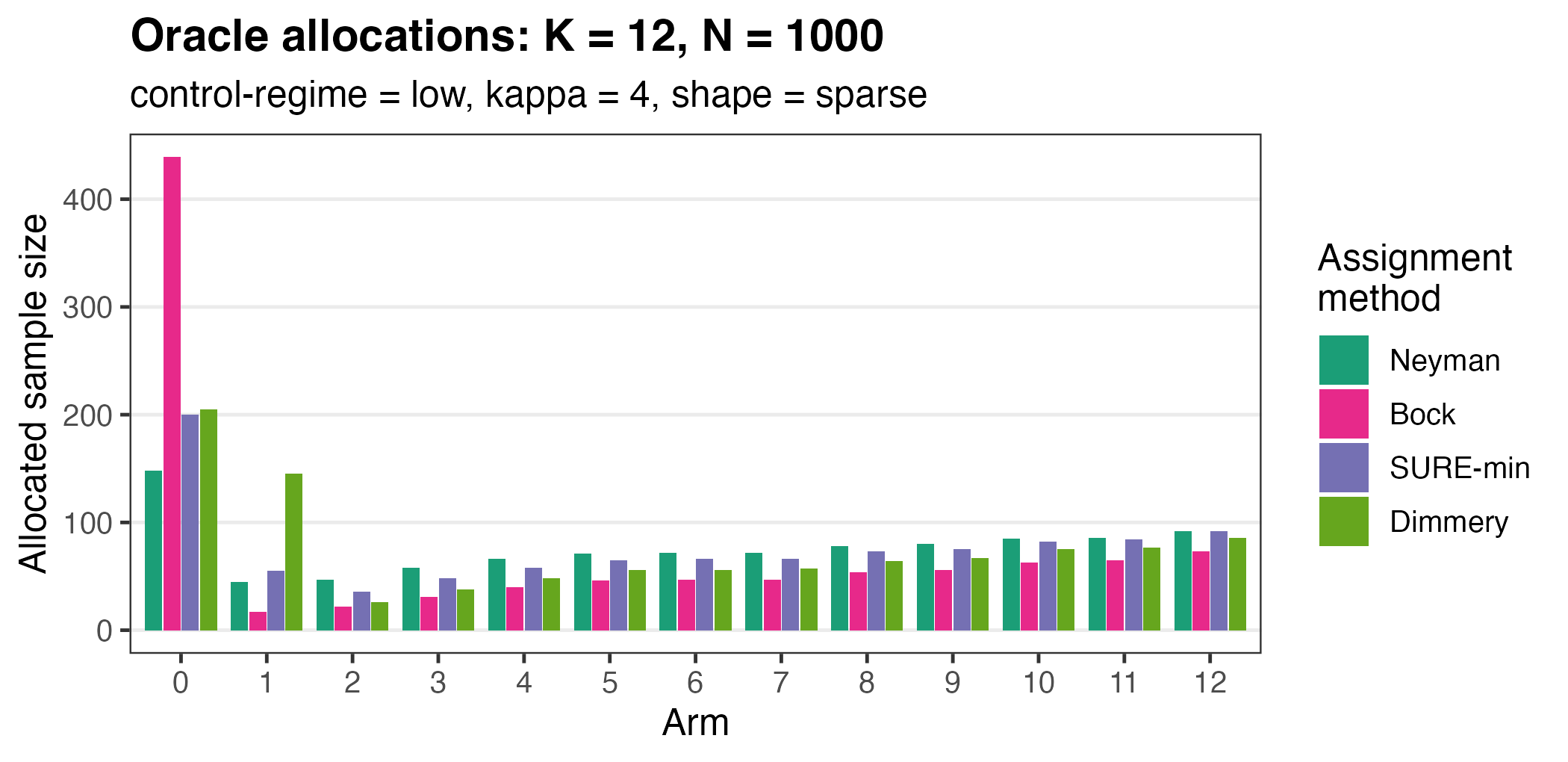}
\caption{Oracle allocations for $K=12$ and $N=1,000$ with low control-variance regime, $\kappa=4$, and sparse $\boldsymbol{\tau}$. Here, we see two key trends. First, the control reallocation for Bock's estimator is much larger than in the otherwise-identical case in which $\bstau$ is dense (Figure \ref{fig:alloc_controlshift_delta9}). This reflects the dependency of the minimizing allocation for Bock's estimator on the signal density. Second, we see that the minimizing allocation for Dimmery's estimator -- but not the Bock or SURE-min estimator -- tends to allocate more units to the treatment arm with strong signal when $\bstau$ is sparse.}
\label{fig:alloc_armspecific_sparse}
\end{figure}

\section{Adaptive Risk Minimization}\label{sec:adaptivity}

The approach in Section \ref{sec:oracle} allows us to characterize the risk-minimizing designs, but it does not provide a pathway to sequential adaptivity because the parameters are assumed known in the design phase. In this section, we introduce a simple algorithm for an adaptive experiment that seeks to minimize the risk of $\bsdeltb, \bsdelts$, or $\bsdeltad$. 

\subsection{Efficient Risk Computation}\label{sec:effRisk}

For each of our candidate estimators, we can use Equation \ref{eq:risk} to obtain the risk. In each case, the resultant expression is a linear combination of expectations of ratios of Gaussian quadratic forms. 

We use $\bsdeltb$ as our motivating example. As discussed in Section \ref{sec:oracle}
\begin{equation}\label{eq:riskStatement}
\begin{aligned}
\mathcal{R}\left( \bsdeltb, \bstau \right) &= \Tr(\boldsymbol \Sigma) + \E \left( \frac{(\tilde p - 2)^2 ||\hbstau||_2^2}{(\hbstau^\tran \bsSig^{-1} \hbstau)^2}  - 
 \frac{2(\tilde p - 2)}{\hbstau^\tran \bsSig^{-1} \hbstau} \Tr(\bsSig) + \frac{4(\tilde p - 2)|| \hbstau ||_2^2}{(\hbstau^\tran \bsSig^{-1} \hbstau)^2} \right)  \\
&= \Tr(\boldsymbol \Sigma) + (\tilde p^2 - 4) \E \left( \frac{||\hbstau||_2^2}{(\hbstau^\tran \bsSig^{-1} \hbstau)^2} \right) - 2(\tilde p - 2)\Tr(\boldsymbol \Sigma) \E \left( \frac{1 }{\hbstau^\top \bsSig^{-1} \hbstau}  \right)  .
\end{aligned}
\end{equation}
\noindent Recall that $\hbstau$ is approximately normally distributed. Hence, the two quantities within expectations are ratios of Gaussian quadratic forms. Results from \cite{bao2013moments} provide a simple procedure to evaluate these expectations via univariate numerical integrals, which can be computed efficiently. 

Observe that we can rewrite the risk as 
\[ \mathcal{R}\left( \bsdeltb, \bstau \right) = \Tr(\boldsymbol \Sigma) + (\tilde p^2 - 4) \E \left( \frac{\bstheta^\tran \boldsymbol \Sigma \bstheta}{\left( \bstheta^{\tran} \bstheta \right)^2} \right) - 2(\tilde p - 2)\Tr(\bsSig) \E \left( \frac{1 }{\left( \bstheta^{\tran} \bstheta \right)}  \right)  \] 
where $\bstheta = \bsSig^{-1/2} \hbstau \sim \mathcal{N} \left( \bsSig^{-1/2} \bstau, \boldsymbol{I_K} \right)$. 
A direct application of the results in \cite{bao2013moments} yields the following integral expressions: 

\begin{equation}\label{eq:integrals1}
\begin{aligned}
\E\left(  \frac{ \bstheta^{\tran}\bsSig \bstheta}{(\bstheta^{\tran} \bstheta)^2} \right) &= \int_0^{\infty}  ( 1 + 2 t)^{-K/2} \cdot \exp\bigg( - (\bstau^\tran \bsSig^{-1} \bstau) \cdot \frac{t}{1 + 2t} \bigg) \cdot \bigg( \frac{\Tr (\bsSig)}{1 + 2t} + \frac{\bstau^\tran \bstau}{(1 + 2t)^2} \bigg) \hspace{1mm} t \hspace{1mm} dt \,, \mbox{ and} \\
\E\left(  \frac{1}{\bstheta^{\tran} \bstheta} \right) &= \int_0^{\infty}( 1 + 2 t)^{-K/2}  \cdot  \exp\bigg( - (\bstau^\tran \bsSig^{-1} \bstau) \cdot \frac{t}{1 + 2t} \bigg) dt .
\end{aligned}
\end{equation}

Crucially, the integrals in \eqref{eq:integrals1} are univariate, and hence can be efficiently computed to high precision. Leveraging a \texttt{C++} integration into \texttt{R} via RCPP \citep{rcpp} yields further speed improvements. In Table \ref{tab:risk_timing}, we provide speed benchmarks in milliseconds over 5,000 repetitions for the computation of the risk of each estimator using RCPP.
All benchmarks were run on a 14-inch MacBook Pro with an Apple M4 Pro processor and 24~GB of memory, running MacOS. Even in high dimensions, the most onerous computation -- the risk of $\bsdeltad$ -- takes only one tenth of a millisecond on average. 

\begin{table}[ht]
\centering
\caption{Avg. computation time (milliseconds) over 5,000 iterations.}
\label{tab:risk_timing}
\begin{tabular}{cccc}
\toprule
Dimension  $K$ & Bock ($\bsdeltb$)  & SURE-min ($\bsdelts$) & Dimmery ($\bsdeltad$) \\
\midrule
4  & 0.0063 & 0.0111 & 0.0187 \\
6  & 0.0064 & 0.0163 & 0.0293 \\
8  & 0.0072 & 0.0177 & 0.0367 \\
12 & 0.0119 & 0.0296 & 0.0630 \\
16 & 0.0163 & 0.0419 & 0.1017 \\
\bottomrule
\end{tabular}
\end{table}

\subsection{Greedy Algorithm}

In the adaptive design literature, there are several different proposals for the cadence at which treatment assignment rules are updated. Many approaches consider a ``pilot" or ``warm-up" phase in which treatment assignments are non-adaptive and equally frequent across arms, with the information collected in this phase used to determine a single assignment rule for the latter experimental phase \citep{blackwell2022batch, cai2024performance, tabord2018stratification, hahn2011adaptive}. Alternatively, treatment assignment rules may be updated in batches \citep{LiOwen2024, zhang2020inference, zhao2023adaptive} or recomputed sequentially with each new arrival \citep{kato2020efficient, dai2023clip}. 

We propose a simple, greedy algorithm that combines several of these approaches. Suppose the adaptive experiment is pre-specified to involve $N$ total units. For the first $N_{\text{warm-up}} \ll N$ arrivals, we assign the treatment according to complete randomization, with equal probability of receiving the control or any of the $K$ active treatments. Denote as $\hbstau^{(i)}$ the causal estimates after arrival $i$ (computed using the difference-in-means estimator), and $\boldsymbol{\hat V^{(i)}}$ the estimated potential outcome variances after arrival $i$ (computed using the standard sample variance formula). Let 
\[ \boldsymbol{n^{(i)}} = \left(n_0^{(i)}, \dots n_K^{(i)} \right)^\tran \in \mathbb{R}^{K + 1}\] 
represent the total number of units assigned to each arm after arrival $i$, and let $\boldsymbol e_k \in \mathbb{R}^{K + 1}$ be the vector whose $(k + 1)^\text{th}$ entry is 1 and whose other entries are $0$. For $i$ such that $N_{\text{warm-up}} < i \leq N$, we assign treatment according to the deterministic rule: 

\[ W_i = \argmin_{k \in \{0,\dots,K\}} \mathcal{R} \left( \bsdelt, \boldsymbol{\hat \tau^{(i-1)}}; \boldsymbol{n^{(i - 1)}} +\boldsymbol e_k,\, \boldsymbol{\hat V^{(i - 1)}} \right). \]

Simply: after the completion of the warm-up period, we assign units to arms based on which assignment will minimize our current estimate of the shrinker risk. This rule can be applied with any of $\bsdelt = \bsdeltb, \bsdelts$, or $\bsdeltad$, and the risk estimates are computed using the efficient RCPP integral representations discussed in Section \ref{sec:effRisk}. Unlike the standard Neyman allocation setting, the shrinker risk depends on both $\bstau$ and $\boldsymbol V$, so we retain updating estimates of both these quantities throughout the experiment. This means early estimates of treatment effects can influence future allocation decisions. Additionally, note that this algorithm is not random after we condition on all the arrivals prior to the $i^\text{th}$ one. 

We do not make claims of optimality for this algorithm relative to alternative treatment rules that might try to minimize the shrinker risk. Nevertheless, we consider this approach attractive in its simplicity and demonstrate its utility in the sections to follow. 

\subsection{Simulations}

We ran simulations to evaluate in which settings adaptivity based on shrinkage provides the most benefit. The simulations consist of a series of experiments where units arrive sequentially with a fixed total sample size $N$. There are two main choices we consider for the simulated experiments: how treatment is allocated to different treatment arms, and after the experiment is over, what estimator is used to estimate the treatment effects. For the treatment allocation approaches, we compare sequential risk-minimization for $\bsdeltb, \bsdelts$, and $\bsdeltad$ to the static allocation methods of complete randomization and Neyman allocation. 

For the estimators, we compare the $\bsdeltb, \bsdelts$, and $\bsdeltad$  to the difference-in-means estimator $\hbstau$. We evaluate the positive-part versions of these estimators, as these are the versions that would typically be used in practice. However, the adaptive allocation rules continue to minimize the surrogate risks derived for the corresponding smooth estimators.

We calculate the mean squared error as each unit arrives in the trial. We simulate repeated iterations, each time drawing a new set of potential outcomes from the same distribution, and running a sequential experiment. 
We summarize the results over $1,000$ iterations by computing the average MSE at each unit's arrival, producing an average MSE trajectory. For detailed information about the simulation design, see Appendix~\ref{sec:sim_design}.

We highlight five regimes among the ones considered in Section~\ref{sec:nonadaptSim}. Figures~\ref{fig:K6-kappa0}-\ref{fig:K6-dense-kappa4} each show the trajectory of mean squared error as units arrive in trial for one simulation regime. All have $K=6$ and high control variance, and we vary the signal-to-noise ratio $\kappa$ from $0$ to $2$ to $4$. When the signal is nonzero, we also show results for both dense and sparse signals $\bstau$. Additional simulations can be found in Appendix~\ref{sec:sim_add_results}; 
Figures~\ref{fig:K12-kappa0}-\ref{fig:K12-dense-kappa4} show results for the $K = 12$ regime, high control-variance regime, while Figures~\ref{fig:K6-low-kappa0}-\ref{fig:K6-low-dense-kappa4} show the $K = 6$ case but for the low control-variance regime. 

Total sample sizes $N$ are selected so that there are $400$ units on average across treatment and control arms, yielding $N = 2,800$ for $K = 6$ and $N = 5,200$ for $K = 12$. \cite{zhao2023adaptive} suggests a sample size of $\sqrt{N}$ for the pilot phase of a two-phase experiment when using an adaptive Neyman Allocation; we opt for a slightly larger value of $N_{\text{warm-up}} = 100$. 


Figures~\ref{fig:K6-kappa0}-\ref{fig:K6-dense-kappa4} demonstrate several trends. First, we observe that using the SURE-min estimator $\bsdelts$ or the Dimmery estimator $\bsdeltad$ yields improvement in average risk over using the difference-in-means estimator $\hbstau$, regardless of the assignment method. These benefits are most noticeable when the SNR $\kappa$ is low; when $\kappa > 0$, they are more noticeable when $\bstau$ is sparse. The improvements are more muted when switching to Bock estimator $\bsdeltb$. The simulation regime -- with $K = 6$ and substantial heterogeneity across $V_0, V_1, \dots, V_K$ -- tends to yield values of the effective dimension parameter $\tilde p$ close to the cutoff value of $2$, restricting the potential estimation improvements from this estimator. 

Beyond the gains from shrinkage estimation itself, assigning treatments to minimize the shrinker risk yields further improvements
over a simple Neyman allocation, when using either $\bsdelts$ or $\bsdeltad$. The effect is most notable when $\kappa$ is small, but persists at modest $\kappa$ values and even at $\kappa = 4$ when $\bstau$ is sparse. In aggregate, the strong empirical performance of the adaptive procedures suggests that minimizing the (analytically tractable) surrogate risk is an effective proxy for minimizing the risk of the corresponding positive-part estimators.


Tables~\ref{tab:adaptiveSimsK6} and ~\ref{tab:adaptiveSimsK12} also summarize the performance of treatment allocation strategies across a wider set of regimes: we vary both the dimension ($K = 6$ and $K = 12)$ and the $V_0$ value (low vs. high). We report two relevant numbers: the average across simulations of the MSE assessed at unit 1,000 in the adaptive trial; and the mean MSE from unit $1,000$ until the end of the experiment, also averaged across simulations. We chose these  measures because the MSE is very large at the outset of the experiment and decreases rapidly as units arrive, so these measures assess trial performance at a more stable, later phase. 

For simplicity, for each treatment allocation strategy we report only the result for the estimator that is most naturally paired with that strategy, e.g. applying Bock's estimator to the risk-minimization treatment allocation that targets Bock's estimator. Both complete randomization and the Neyman treatment allocation strategies are paired with the difference-in-means estimator. For each of the two metrics, the value of the lowest MSE is in bold.


Tables~\ref{tab:adaptiveSimsK6} and ~\ref{tab:adaptiveSimsK12} demonstrate that the risk-minimization treatment allocations targeting the SURE-min and Dimmery's estimator provide the best performance across all regimes. 
Targeting the SURE-min estimator results in the best performance when $\bstau$ is dense, while targeting Dimmery's estimator performs better when $\bstau$ is sparse. When $\bstau$ is sparse, the difference between targeting the SURE-min estimator and Dimmery's estimator is largest when $\kappa = 4$.


\begin{figure}[p]
    \centering
    \includegraphics[width=0.9\linewidth]{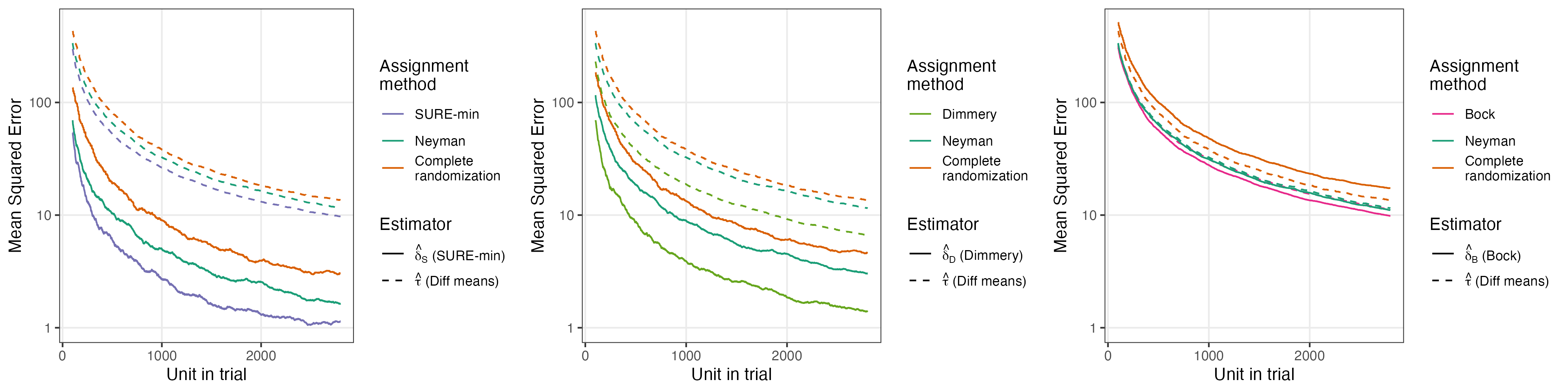}
    \caption{$K = 6, \kappa=0$.}
    \label{fig:K6-kappa0}
    \includegraphics[width=0.9\linewidth]{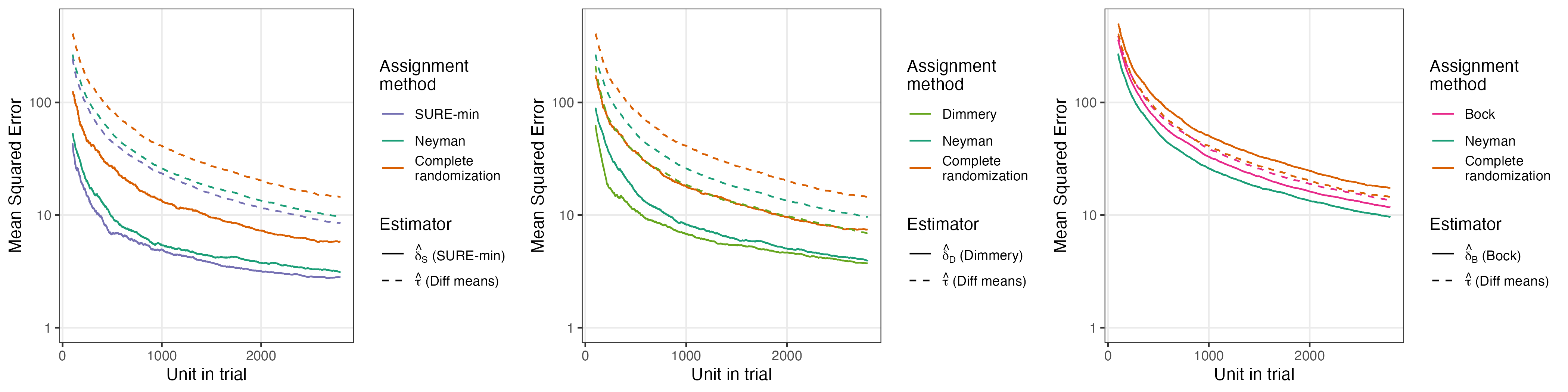}
    \caption{$K = 6, \bstau$ sparse, $\kappa=2$.}
    \label{fig:K6-sparse-kappa2}
    \includegraphics[width=0.9\linewidth]{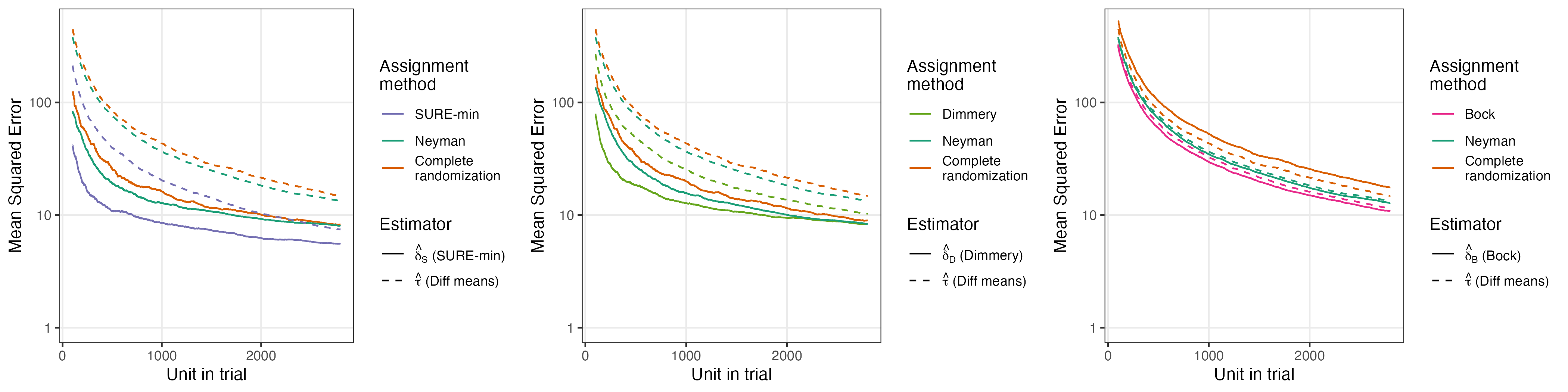}
    \caption{$K = 6, \bstau$ dense, $\kappa=2$.}
    \label{fig:K6-dense-kappa2}
\end{figure}

\begin{figure}[p]
    \centering
    \includegraphics[width=0.9\linewidth]{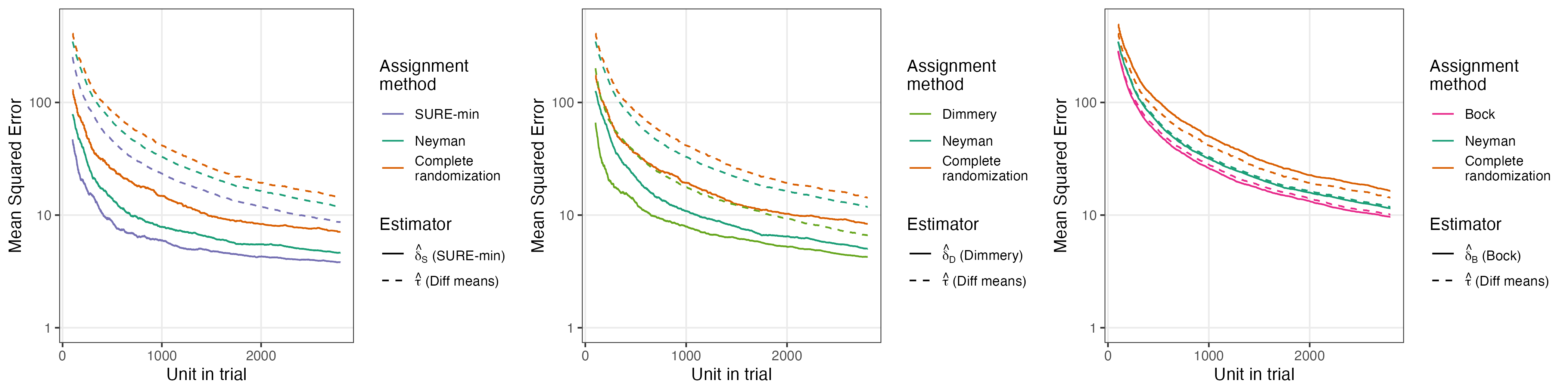}
    \caption{$K = 6, \bstau$ sparse, $\kappa=4$.}
    \label{fig:K6-sparse-kappa4}
    \includegraphics[width=0.9\linewidth]{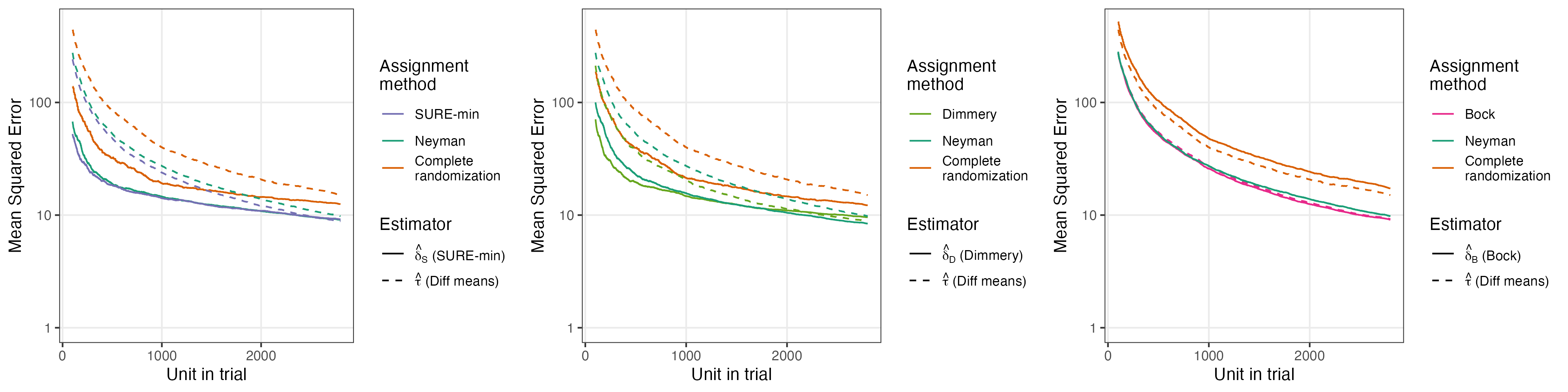}
    \caption{$K = 6, \bstau$ dense, $\kappa=4$.}
    \label{fig:K6-dense-kappa4}
\end{figure}

\begin{table}[p]
\centering
\begin{threeparttable}
\caption{Summary of simulation performance of adaptive treatment allocation strategies for $K = 6$}
\label{tab:adaptiveSimsK6}

\small
\setlength{\tabcolsep}{4pt}

\begin{tabular}{rlll|rrrrr|rrrrr}
\toprule
 &  &  & 
& \multicolumn{5}{c}{MSE at Unit $1000$}
& \multicolumn{5}{c}{Mean MSE over Unit $\geq 1000$} \\
\cmidrule(lr){5-9} \cmidrule(lr){10-14}
$K$& $V_0$ regime & $\bstau$ shape & $\kappa$ & CR & Neyman & $\bsdeltb$ & $\bsdelts$ & $\bsdeltad$
& CR & Neyman & $\bsdeltb$ & $\bsdelts$ & $\bsdeltad$ \\
\midrule
6 & high & dense & 0 & 38.3 & 32.6 & 27.9 & \textbf{2.7} & 3.9 & 21.3 & 18.5 & 15.7 & \textbf{1.5} & 2.2\\
6 & high & sparse & 0 & 38.3 & 25.6 & 23.8 & \textbf{2.7} & 3.5 & 22.6 & 14.5 & 13.4 & \textbf{1.6} & 2.0\\
6 & high & dense & 2 & 43.4 & 36.7 & 29.7 & \textbf{8.5} & 12.7 & 24.4 & 21.2 & 17.1 & \textbf{6.7} & 10.0\\
6 & high & sparse & 2 & 41.3 & 25.9 & 33.0 & \textbf{5.0} & 6.8 & 23.3 & 15.3 & 18.7 & \textbf{3.4} & 4.9\\
6 & high & dense & 4 & 40.1 & 27.3 & 25.8 & \textbf{14.1} & 14.8 & 23.7 & 15.7 & 14.5 & \textbf{11.3} & 11.5\\
6 & high & sparse & 4 & 41.6 & 32.9 & 26.1 & \textbf{6.0} & 7.9 & 22.9 & 18.7 & 14.9 & \textbf{4.5} & 5.6\\
\midrule
6 & low & dense & 0 & 30.1 & 17.1 & 11.4 & \textbf{1.5} & 2.6 & 17.0 & 9.9 & 6.6 & \textbf{0.8} & 1.4\\
6 & low & sparse & 0 & 21.6 & 16.6 & 11.3 & \textbf{1.1} & 2.0 & 12.4 & 9.5 & 6.3 & \textbf{0.6} & 1.2\\
6 & low & dense & 2 & 18.2 & 20.0 & 14.7 & \textbf{4.9} & 5.2 & 10.3 & 11.3 & 8.8 & \textbf{3.9} & 4.3\\
6 & low & sparse & 2 & 26.0 & 23.5 & 11.8 & \textbf{3.0} & 4.9 & 15.0 & 13.7 & 6.8 & \textbf{2.4} & 3.7\\
6 & low & dense & 4 & 14.2 & 15.3 & 13.1 & \textbf{6.9} & 9.3 & 8.1 & 8.8 & 7.8 & \textbf{5.6} & 7.1\\
6 & low & sparse & 4 & 18.0 & 17.7 & 13.3 & \textbf{4.7} & 5.2 & 10.5 & 10.1 & 8.0 & 3.8 & \textbf{3.7}\\
\bottomrule
\end{tabular}

\footnotesize
Each metric (MSE for Unit 1000 and Mean MSE for Unit
$\geq 1000$) is the mean across simulation iterations.
For each metric, the value with the lowest unrounded MSE is in bold.
CR stands for Complete Randomization. $\bsdeltb$ targets the Bock estimator, $\bsdelts$ targets the SURE-min estimator, and $\bsdeltad$ targets the Dimmery estimator.

\end{threeparttable}
\end{table}

\begin{table}[p]
\centering
\begin{threeparttable}
\caption{Summary of simulation performance of adaptive treatment allocation strategies for $K = 12$}
\label{tab:adaptiveSimsK12}

\small
\setlength{\tabcolsep}{4pt}
\begin{tabular}{rlll|rrrrr|rrrrr}
\toprule
 &  &  & 
& \multicolumn{5}{c}{MSE at Unit $1000$}
& \multicolumn{5}{c}{Mean MSE over Unit $\geq 1000$} \\
\cmidrule(lr){5-9} \cmidrule(lr){10-14}
$K$& $V_0$ regime & $\bstau$ shape & $\kappa$ & CR & Neyman & $\bsdeltb$ & $\bsdelts$ & $\bsdeltad$
& CR & Neyman & $\bsdeltb$ & $\bsdelts$ & $\bsdeltad$ \\
\midrule
12 & high & dense & 0 & 169.4 & 102.8 & 94.8 & 7.4 & \textbf{2.8} & 65.3 & 40.6 & 34.9 & 2.8 & \textbf{1.2}\\
12 & high & sparse & 0 & 157.3 & 87.7 & 78.7 & 5.9 & \textbf{4.3} & 60.5 & 33.9 & 30.2 & 2.2 & \textbf{1.7}\\
12 & high & dense & 2 & 166.5 & 100.5 & 89.2 & 16.9 & \textbf{14.0} & 62.7 & 39.3 & 34.2 & 11.2 & \textbf{11.1}\\
12 & high & sparse & 2 & 179.9 & 100.9 & 94.4 & 7.9 & \textbf{5.8} & 72.3 & 39.6 & 34.6 & 3.6 & \textbf{3.0}\\
12 & high & dense & 4 & 164.5 & 97.2 & 86.4 & 24.3 & \textbf{23.5} & 65.3 & 37.8 & 33.1 & \textbf{17.3} & 18.6\\
12 & high & sparse & 4 & 178.6 & 73.5 & 90.2 & 12.2 & \textbf{7.5} & 68.9 & 29.2 & 33.7 & 8.0 & \textbf{4.7}\\
\midrule
12 & low & dense & 0 & 71.2 & 58.4 & 55.0 & 3.4 & \textbf{1.9} & 27.7 & 22.9 & 18.6 & 1.3 & \textbf{0.8}\\
12 & low & sparse & 0 & 62.5 & 52.9 & 63.0 & 3.4 & \textbf{1.6} & 24.2 & 20.9 & 21.3 & 1.2 & \textbf{0.6}\\
12 & low & dense & 2 & 73.0 & 72.3 & 48.6 & 7.9 & \textbf{6.0} & 28.6 & 29.2 & 17.5 & 5.6 & \textbf{5.0}\\
12 & low & sparse & 2 & 65.6 & 77.0 & 54.9 & 5.4 & \textbf{3.2} & 25.8 & 30.0 & 19.6 & 3.4 & \textbf{1.9}\\
12 & low & dense & 4 & 65.9 & 60.4 & 54.9 & 14.0 & \textbf{10.4} & 25.5 & 23.2 & 21.1 & 10.4 & \textbf{8.8}\\
12 & low & sparse & 4 & 66.5 & 65.6 & 46.7 & 4.6 & \textbf{2.8} & 26.2 & 25.2 & 16.7 & 2.5 & \textbf{1.6}\\
\bottomrule
\end{tabular}

\end{threeparttable}
\end{table}

\section{Confidence Intervals}\label{sec:cis}

Practitioners often require uncertainty intervals that can be reported along with shrinker point estimates. We briefly discuss the construction of shortened intervals that can be reported at the conclusion of these adaptive trials.

Our confidence intervals are based on the FAB (``frequentist assisted by Bayes") procedure, first proposed in \cite{pratt1963shorter} and developed into a modern framework in \cite{yu_hoff_2018}. 
Pratt observed that, when we have access to prior information about the value of a parameter, we can construct confidence intervals for the parameter that attain frequentist coverage but are typically shorter than Wald intervals. These intervals are narrower than Wald intervals when the observed data is in a high-density region of the prior and wider than Wald intervals when the observed data is in the tails of the prior, generating shorter intervals ``on average" when integrated over the prior.

\cite{yu_hoff_2018} extend this approach to the multigroup data setting. In the absence of an assumed prior, they assume a model for the heterogeneity across the groups, and estimates its parameters using data from all of the groups. They also extend Pratt's approach to the case where group-specific variances are unknown, and show the approach generalizes straightforwardly to the heteroscedastic case. Intervals are extended to linear regression coefficients in \cite{10.1214/18-EJS1517} and to the small-area estimation problem in \cite{10.1093/jssam/smz010}.

To account for the heteroscedastic and dependent sampling model described in Equation~\eqref{eq:samplingModel}, we make some minor adaptations to the procedure suggested in \cite{yu_hoff_2018}. Details are given in Appendix \ref{appendix:fab}. We suggest a single FAB interval structure, regardless of the estimator that is being used. 

In Table \ref{tab:k6-fab-high-hetero-high-v}, we show simulation results for the $K = 6$, high-control-variance setting. Both Wald and adapted FAB intervals are computed for each of the estimated causal effects at the end of each simulated trial. We again vary the signal-to-noise ratio $\kappa$; the signal density; and the adaptive allocation scheme used to assign treatments. Minimum coverage rates across the treatment effects are provided in the fourth and fifth columns. Notably, minimum coverage remains close to nominal levels in all settings. Moreover, the final column shows that we can typically attain substantial length reductions (12-15\%), relative to Wald intervals, using the FAB approach. Further simulations for the $K = 12$ case can be found in Appendix \ref{appendix:fab}. 

\begin{table}[ht]
\centering
\begin{tabular}{lllrrr}
\toprule
\makecell[l]{Signal-to-\\Noise Ratio} & $\bstau$ setting & \makecell[l]{Allocation \\Scheme} & \makecell[l]{Minimum \\ Wald Coverage} & \makecell[l]{Minimum \\ FAB Coverage} & \makecell[l]{Avg. FAB\\length reduction} \\
\midrule
$\kappa=0$ & $-$ & Balanced & 94\% & 94\% & $-$15\% \\
 & & Neyman & 94\% & 94\% & $-$15\% \\
 & & SURE-min & 95\% & 95\% & $-$15\% \\
 & & Dimmery & 95\% & 95\% & $-$15\% \\
 & & Bock & 94\% & 94\% & $-$15\% \\
\addlinespace
$\kappa=2$ & sparse & Balanced & 94\% & 93\% & $-$15\% \\
 & & Neyman & 94\% & 94\% & $-$15\% \\
 & & SURE-min & 95\% & 95\% & $-$15\% \\
 & & Dimmery & 95\% & 95\% & $-$15\% \\
 & & Bock & 94\% & 94\% & $-$14\% \\
\addlinespace
$\kappa=2$ & dense & Balanced & 94\% & 93\% & $-$14\% \\
 & & Neyman & 94\% & 93\% & $-$13\% \\
 & & SURE-min & 96\% & 95\% & $-$14\% \\
 & & Dimmery & 96\% & 95\% & $-$14\% \\
 & & Bock & 94\% & 93\% & $-$13\% \\
\addlinespace
$\kappa=4$ & sparse & Balanced & 94\% & 94\% & $-$14\% \\
 & & Neyman & 94\% & 94\% & $-$14\% \\
 & & SURE-min & 95\% & 95\% & $-$14\% \\
 & & Dimmery & 95\% & 94\% & $-$14\% \\
 & & Bock & 93\% & 93\% & $-$14\% \\
\addlinespace
$\kappa=4$ & dense & Balanced & 94\% & 92\% & $-$13\% \\
 & & Neyman & 95\% & 93\% & $-$12\% \\
 & & SURE-min & 95\% & 94\% & $-$12\% \\
 & & Dimmery & 96\% & 95\% & $-$13\% \\
 & & Bock & 94\% & 93\% & $-$12\% \\
\bottomrule
\end{tabular}
\caption{FAB interval performance for $K=6$ with high control variance. Coverage columns report the minimum across the $K$ treatment effects; length reduction is computed as the mean FAB interval length divided by the mean Wald interval length, minus one.}
\label{tab:k6-fab-high-hetero-high-v}
\end{table}

We note two related problems reserved for future work. First, we do not attempt to construct ``always-valid" or ``anytime-valid" confidence sets. Such intervals -- which have been a major focus of research in online A/B testing \citep[see e.g.][]{Johari2017, johari2022always, maharaj2023anytime} -- are designed to address the ``peeking" problem, in which analysts use inferential quantities computed midway through an online trial to assess whether to terminate the trial. If standard confidence sets are used, this will inflate the Type I error, while always-valid intervals are designed to be consulted at any point in the trial without inflating the false positive rate. Here, we focus solely on confidence intervals that are constructed at the end of the trial. 

Second, we do not make an explicit correction for adaptivity in the trial. Several papers \citep{hadad2021confidence, zhao2023adaptive} have highlighted that adaptivity in data collection induces dependency of later treatment assignments on earlier outcomes, which can induce bias in standard causal estimators. This dependency can also make the estimators' sampling distributions heavy-tailed, undermining the validity of standard confidence procedures. Here, our estimators are already known to be biased. Moreover, in simulations, we do not find that our confidence procedures systematically undercover due to adaptivity. We note that specialized corrections may be needed in cases involving more extreme heteroscedasticity.

\section{Discussion}\label{sec:conclusion}

We have proposed an approach to adaptive experimental designs, in which our causal estimates are obtained by applying an Empirical Bayes shrinker to the estimated contrasts between active treatment arms and a control arm. We considered three shrinkers targeted to the case of heteroscedastic and dependent data: Bock's estimator $\bsdeltb$, an easy-to-compute SURE-minimizing estimator $\bsdelts$, and Dimmery's estimator $\bsdeltad$. Each admits a testable condition under which its risk is guaranteed to be smaller than that of the difference-in-means estimator $\hbstau$.  

Under a non-adaptive design with known parameter values, the risk-minimizing allocation for each of these shrinkage estimators diverges from a standard Neyman allocation. To varying degrees, these allocations inflate the number of units assigned to the control arm, thus de-correlating the arm-specific causal estimates. This makes intuitive sense, as it provides the shrinkers with more independent data from which to learn. 

We also showed that these estimators can be used for sequential adaptivity, because their risks can be expressed as efficiently-computable univariate integrals. 
Hence, by maintaining online estimates of the mean and variance parameters and using a greedy minimization scheme, we can assign units to treatment arms to minimize the estimated risk of each estimator. Simulations demonstrate that this method outperforms an adaptive Neyman allocation, particularly when the signal-to-noise ratio is small. 

There are many future directions for this research. First, our greedy algorithm -- used for adaptively assigning treatments to units -- is deterministic conditional on the units already observed. This approach makes it sensitive to unlucky observations early in the experiment. Many modern methods \citep[e.g.]{dai2023clip, LiOwen2024} instead allocate treatment probabilistically at each stage. These methods either gradually update treatment probabilities with each new arrival to the trial, or recompute treatment probabilities at each stage in a batched experiment. Such approaches may better manage the exploration-exploitation tradeoff than our greedy approach, and may also yield faster convergence to the true optimal treatment probabilities for a given value of $\bstau$ and the potential outcome variances $\{V_k\}_{k = 0}^K$. 

Moreover, we have considered three reasonable estimators in this manuscript, finding that the SURE-min estimator $\bsdelts$ and Dimmery's estimator $\bsdeltad$ tend to yield improved performance in different simulation regimes. Yet our method is valid for any shrinkage estimator whose risk can be expressed via a sum of expectations of Gaussian quadratic form ratios. There is a rich literature on the construction of shrinkage estimators using unbiased risk estimators in the vein of Stein's Unbiased Risk Estimate \citep{li1985stein, xie2012estimating, donoho1995adapting}. Hence, future efforts should focus on the design of shrinkers that are specially tailored to the tasks of causal estimation in adaptive experiments. Such estimators would be especially useful if they incorporated covariates to improve estimation accuracy. Construction of tailored shrinkers -- in tandem with practical guidance about when and where to use different candidate estimators -- would help to operationalize our method and make it accessible to a broader group of practitioners. 

\bibliography{references}

\appendix
\addtocontents{toc}{\protect\setcounter{tocdepth}{-1}}

\section{Proof of Lemma \ref{lemma:sure}}\label{sec:proofSureLemma}

\begin{proof}
The proof draws heavily from Lemma 4.1 and Theorem 4.1 in \cite{strawderman2003minimax}. Because $\bsSigma$ is symmetric and positive definite, it has a unique symmetric positive definite square root $\bsSigma^{1/2}$. Define $\bsZ = \bsSigma^{-1/2} \bsX$ (and, equivalently, $\bsSigma^{1/2}  \bsZ = \bsX$) such that $\bsZ \sim \mathcal{N} \left( \bsSigma^{-1/2}  \bsmu, \boldsymbol I \right)$. Observe:

\begin{align*}
\e \left( || \bsdelt(\bsX) - \bsmu ||_2^2 \right) &= \e \left( \left( \bsX - g(\bsX) - \bsmu \right)^\tran \left( \bsX - g(\bsX) - \bsmu \right) \right) \\
&= \e \left( \left( \bsSigma^{1/2}  \bsZ - g(\bsSigma^{1/2}  \bsZ) - \bsmu \right)^\tran \left( \bsSigma^{1/2}  \bsZ - g(\bsSigma^{1/2}  \bsZ) - \bsmu \right) \right) \\
&= \e \left( \left( \bsZ - \bsSigma^{-1/2} g(\bsSigma^{1/2}  \bsZ) - \bsSigma^{-1/2} \bsmu \right)^\tran \bsSigma \left(   \bsZ - \bsSigma^{-1/2} g(\bsSigma^{1/2}  \bsZ) - \bsSigma^{-1/2} \bsmu \right) \right) 
\end{align*}

The final line is the expected covariance-weighted squared-error loss when estimating the mean of $\bsZ$ using 
\[ \bsdelt'(\bsZ) = \bsZ - f(\bsZ) \hspace{5mm} \text{ for } \hspace{5mm} f(\bsZ) = \bsSigma^{-1/2} g(\bsSigma^{1/2} \bsZ), \]
where $\bsZ$ is multivariate normal with identity covariance. Hence, we are back to the standard setting for using SURE. Applying SURE, we obtain 
\begin{align*}
\e \left( || \bsdelt(\bsX) - \bsmu ||_2^2 \right) &= 
\e \left( || \bsSig^{1/2} \left(\bsdelt'(\bsZ) - \bsSigma^{-1/2} \bsmu\right) ||_2^2 \right) \\
&= \Tr(\bsSig) + \e \left( || \bsSig^{1/2} f(\bsZ) ||_2^2 \right) -
2\,\operatorname{tr}\!\left(\bsSig \mathcal{J}_f(\bsZ)\right),
\end{align*}
where $\mathcal{J}_{f}( \bsZ )$ is the Jacobian matrix of $f(\cdot)$ evaluated at $\bsZ$. We observe immediately that 
\begin{align*}
|| \bsSig^{1/2} f(\bsZ) ||_2^2 &= g(\bsSigma^{1/2} \bsZ)^\tran  \bsSigma^{-1/2} \bsSigma^{1/2} \bsSigma^{1/2} \bsSigma^{-1/2} g(\bsSigma^{1/2} \bsZ) \\
&= g(\bsSigma^{1/2} \bsZ)^\tran g(\bsSigma^{1/2} \bsZ) = ||g(\bsX)||_2^2.
\end{align*}
Lastly, observe that 
\begin{align*}
\Tr(\bsSig \mathcal{J}_f(\bsZ)) &= \Tr \left(\bsSig \bsSig^{-1/2}
\mathcal{J}_g(\bsSig^{1/2}\bsZ)
\bsSig^{1/2} \right) = \Tr \left(\bsSig 
\mathcal{J}_g(\bsX) \right),
\end{align*}
where we have used the cyclic property of the trace and the fact that $\bsSigma^{1/2} \bsZ = \bsX$.

\end{proof}

\section{Proof of Lemma \ref{lemma:riskBound_deltas}}\label{sec:proof_deltas}

\begin{proof}
We compute the risk of $\bsdelts$ using Expression \ref{eq:risk}:
\begin{align*}
\mathcal{R}\left( \bsdelts, \bstau \right) 
&=  \Tr(\bsSigma) + \E \left( \frac{\Tr(\bsSig)^2}{|| \hbstau||_2^2} - 2\frac{\Tr(\bsSig)^2}{|| \hbstau||_2^2} + 4 \frac{\Tr(\bsSig) \left(\hbstau^ \tran \bsSig \hbstau \right)}{||\hbstau||_2^4}   \right) \\
&=  \Tr(\bsSigma) + \E \left( - \frac{\Tr(\bsSig)^2}{|| \hbstau||_2^2} + 4 \frac{\Tr(\bsSig) \left(\hbstau^ \tran \bsSig \hbstau \right)}{||\hbstau||_2^4}   \right) \\
&\leq  \Tr(\boldsymbol \Sigma) + \E \left( \frac{4\cdot \Tr(\bsSig) \cdot\lambdamax}{|| \hbstau||_2^2}  - \frac{\Tr(\boldsymbol \Sigma)^2 }{|| \hbstau||_2^2}  \right) \\
&=  \Tr(\boldsymbol \Sigma) + \Tr(\bsSig) \cdot \E \left( \frac{4 \cdot \lambdamax  - \Tr(\bsSigma)}{|| \hbstau||_2^2}    \right). 
\end{align*}

We know that $\mathcal{R}\left( \hbstau, \bstau \right) = \Tr(\bsSigma)$, so $\bsdelts$ dominates $\hbstau$ if the expectation term is negative, which occurs precisely when the given condition holds. 

\end{proof}

\section{Proof of Lemma \ref{lemma:riskBound_deltad}}\label{sec:proof_deltad}
\begin{proof}
From Expression \ref{eq:risk}, we can directly compute the risk of $\bsdeltad$. Denote as $\bsSigma_{\star}$ the diagonal matrix whose diagonal entries are equal to the diagonal entries of $\bsSigma$. Then, we obtain
\begin{align*}
\mathcal{R}\left( \bsdeltad, \bstau \right) 
&=  \Tr(\bsSigma) + \E \left( \frac{(K - 2)^2}{|| \hbstau||_2^4} \left(\hbstau^\tran \bsSigma_{\star}^2 \hbstau \right) + \frac{4(K - 2)}{|| \hbstau||_2^4} \left(\hbstau^\tran \bsSigma \bsSigma_{\star} \hbstau \right) - \frac{2(K-2) \Tr(\boldsymbol \Sigma_{\star}^2) }{|| \hbstau||_2^2}  \right) \\
&\leq \Tr(\bsSigma) + (K - 2) \E \left( \frac{4 \hbstau^\tran \bsSigma \bsSigma_{\star} \hbstau }{|| \hbstau||_2^4} + \frac{(K - 2) \max_k \sigma_k^4 - 2\cdot\Tr(\bsSigma_{\star}^2)}{||\hbstau||_2^2} \right)
\end{align*}

We bound the first quadratic form directly. For any vector $\hbstau$,
\begin{align*}
\hbstau^\tran \bsSigma \bsSigma_{\star} \hbstau
&= \hbstau^\tran \bsSigma \big( \bsSigma_{\star} \hbstau \big) \\
&\le \| \hbstau \|_2 \, \big\| \bsSigma \big( \bsSigma_{\star} \hbstau \big) \big\|_2 \\
&\le \| \hbstau \|_2 \, \lambdamax \, \big\| \bsSigma_{\star} \hbstau \big\|_2 \\
&\le \lambdamax \, \big( \max_k \sigma_k^2 \big) \, \| \hbstau \|_2^2 .
\end{align*}

Plugging this in, we obtain
\begin{align*}
\mathcal{R}\left( \bsdeltad, \bstau \right) 
&\leq \Tr(\bsSigma) + (K - 2) \E \left(  \frac{4 \lambdamax \big( \max_k \sigma_k^2 \big)  + (K - 2) \max_k \sigma_k^4 - 2\cdot\Tr(\bsSigma_{\star}^2)}{||\hbstau||_2^2} \right)
\end{align*}
If the numerator is negative, then $\bsdeltad$ dominates $\hbstau$. This yields precisely the given condition. 

\end{proof}

\section{Proof of Lemma \ref{lemma:neymanAlloc}}\label{sec:proof_neymanAlloc}

\begin{proof}
The objective is convex, and an equal allocation is one example of a feasible interior point satisfying Slater's condition. Hence, we can use the standard approach of Lagrange Multipliers to find the optimum. The Lagrangian is given by

\begin{align*}
\mathcal{L}(n_0, \dots, n_K, \lambda) &= \frac{K V_0}{n_0} + \sum_{k = 1}^K \frac{V_k}{n_k} + \lambda \left( \sum_{k = 0}^K n_k - N \right),
\end{align*}
with derivatives 
\begin{align*}
\frac{\partial \mathcal{L}(n_0, \dots, n_K, \lambda) }{\partial n_0} &= - \frac{K V_0}{n_0^2} + \lambda,\\
\frac{\partial \mathcal{L}(n_0, \dots, n_K, \lambda) }{\partial n_k} &= - \frac{V_k}{n_k^2} + \lambda, \hspace{10mm} k = 1, \dots, K \\
\frac{\partial \mathcal{L}(n_0, \dots, n_K), \lambda }{\partial \lambda} &= \sum_{k = 0}^K n_k - N. 
\end{align*}
Setting these derivatives equal to 0, we obtain
\begin{equation}\label{eq:eqDeriv}
\lambda = \frac{K V_0}{n_0^2} = \frac{V_1}{n_1^2} = \dots = \frac{V_K}{n_K^2}.
\end{equation}
This implies that 
\begin{align*}
\sum_{k = 0}^K n_k = \frac{1}{\sqrt{\lambda}} \left( \sqrt{K V_0} + \sum_{k = 1}^K \sqrt{V_k} \right)  = N. 
\end{align*}
from which the result directly follows. 
\end{proof}

\section{Derivation of Approximation (\ref{eq:solApprox})}\label{appendix:solApprox}

Here, we use the notation $\bsSig(\bsn)$ and $\Delta(\bsn)$ to make clear that both these objects are functions of the sample sizes. We can write the Lagrangian for Problem \ref{eq:opt-prob-enhanced} as
\begin{align*}
\mathcal{L}(\bsn, \lambda) &= \Tr\big(\bsSig(\bsn)\big) + \Delta (\bsn) + \lambda \left( \boldsymbol 1^\top \bsn \right).
\end{align*}
with gradient 
\[ \nabla_{\bsn}\mathcal{L}(\bsn, \lambda) = \nabla_{\bsn} \Tr\big(\bsSig(\bsn)\big) + \nabla_{\bsn} \Delta (\bsn) + \lambda \boldsymbol 1^\top.
\]

We consider the behavior of the Lagrangian at $\bsns + \bsd$, a point that is ``local" to the Neyman allocation. Under the assumption that the regularization term and its curvature are small, we suppose 
\[ \left. \nabla_{\bsn} \Delta (\bsn) \right|_{\bsn=\bsns + \bsd} \approx \left. \nabla_{\bsn} \Delta (\bsns)\right|_{\bsn=\bsns}.\] 
Next, we Taylor expand the trace term around $\boldsymbol{n^\star}$, obtaining
\begin{align*}
\left.\nabla_{\bsn}\mathcal{L}(\bsn, \lambda)\right|_{\bsns + \bsd} & \approx \left. \nabla_{\bsn} \Tr\big(\bsSig(\bsn)\big)\right|_{\bsns} + \left. \nabla_{\bsn}^2 \Tr\big(\bsSig(\bsn)\big)\right|_{\bsns} \bsd +  \left. \nabla_{\bsn} \Delta (\bsn)\right|_{\bsns} + \lambda \boldsymbol 1^\top.
\end{align*}

The solution to Problem \ref{eq:opt-prob} is $\bsns$, with associated Lagrange multiplier $\lambda^\star$. The first-order condition to Problem \ref{eq:opt-prob} implies 
\begin{equation}\label{eq:foc_base}
\left. \nabla_{\bsn} \Tr\big(\bsSig(\bsn)\big)\right|_{\bsns} + \lambda^\star \boldsymbol 1^\top = \boldsymbol 0^\top.
\end{equation}
Denote the solution to Problem \ref{eq:opt-prob-enhanced} as $\bsnd = \bsns + \bsdd$, with associated Lagrange multiplier $\lambda^\dag$. The first-order condition for the  solution to Problem \ref{eq:opt-prob-enhanced} is
\begin{align*}
\left.\nabla_{\bsn}\mathcal{L}(\bsn, \lambda)\right|_{\bsns + \bsdd} \approx \left. \nabla_{\bsn} \Tr\big(\bsSig(\bsn)\big)\right|_{\bsns} + \left. \nabla_{\bsn}^2 \Tr\big(\bsSig(\bsn)\big)\right|_{\bsns} \bsdd +  \left. \nabla_{\bsn} \Delta (\bsn)\right|_{\bsns} + \lambda^\dag \boldsymbol 1^\top = \boldsymbol 0^\top.
\end{align*}
Plugging in Equation~\eqref{eq:foc_base}, we obtain 
\begin{align*}
\left. \nabla_{\bsn}^2 \Tr\big(\bsSig(\bsn)\big)\right|_{\bsns} \bsdd +  \left. \nabla_{\bsn} \Delta (\bsn)\right|_{\bsns} + \left(\lambda^\dag - \lambda^\star \right) \boldsymbol 1^\top &= \boldsymbol 0^\top. \\
\left. \nabla_{\bsn}^2 \Tr\big(\bsSig(\bsn)\big)\right|_{\bsns} \bsdd &= \left(\lambda^\star - \lambda^\dag \right) \boldsymbol 1^\top - \left. \nabla_{\bsn} \Delta (\bsn)\right|_{\bsns} \\
\bsdd &= \left( \left. \nabla_{\bsn}^2 \Tr\big(\bsSig(\bsn)\big)\right|_{\bsns} \right)^{-1} \bigg(\left(\lambda^\star - \lambda^\dag \right) \boldsymbol 1^\top - \left. \nabla_{\bsn} \Delta (\bsn)\right|_{\bsns} \bigg).
\end{align*}

Now, we consider the Hessian term. Recall that $\Tr\big(\bsSig(\bsn)\big) = \frac{KV_0}{n_0}+\sum_{k=1}^K\frac{V_k}{n_k}$. Differentiation yields
\[  \frac{\partial^2}{\partial n_0^2} \Tr\big(\bsSig(\bsn)\big) = \frac{2KV_0}{(n_0)^3} \hspace{5mm} \text{ and } \hspace{5mm} \frac{\partial^2}{\partial n_k^2} \Tr\big(\bsSig(\bsn)\big) = \frac{2V_k}{(n_k)^3}, \qquad k=1,\dots,K.\]
All mixed partial derivatives are  $0$. Plugging in the Neyman allocation solution, we obtain 
\[ \left( \left. \nabla_{\bsn}^2 \Tr\big(\bsSig(\bsn)\big) \right|_{\bsns} \right)^{-1} = \frac{1}{2\sna^2} \operatorname{diag}(n_0^\star,\dots,n_K^\star).\]
Thus, on a coordinate-wise basis, we obtain 
\begin{equation}\label{eq:dkIntermediate}
d_k^\dag = \frac{n_k^\star}{2\sna^2} \left[ \lambda^\star-\lambda^\dag - \left. \frac{\partial\Delta}{\partial n_k} \right|_{\bsns} \right].
\end{equation}

Next, we observe that $\boldsymbol 1^\top \bsdd = 0$ because both solutions respect the sample-size constraint. Hence, we can sum up Equation~\eqref{eq:dkIntermediate} over all the coordinates, to obtain 
\begin{align*}
0 = \sum_{k=0}^K d_k^\dag &= \frac{1}{2S_{\mathrm{NA}}^2} \sum_{k=0}^K n_k^\star \left[ \lambda^\star-\lambda^\dag - \left. \frac{\partial\Delta}{\partial n_k} \right|_{\bsns} \right] \\
\lambda^\star-\lambda^\dag &= \frac{1}{N} \sum_{j=0}^K n_j^\star \left. \frac{\partial\Delta}{\partial n_j} \right|_{\bsns}.
\end{align*}

Substituting this expression into Equation~\eqref{eq:dkIntermediate} and dividing by $n_k^\star$, we obtain 
\[ \frac{d_k}{n_k^\star} = \frac{1}{2 \sna^2} \left[ \frac{1}{N} \sum_{j=0}^K n_j^\star \left. \frac{\partial\Delta}{\partial n_j} \right|_{\bsns} - \left. \frac{\partial\Delta}{\partial n_k} \right|_{\bsn=\bsns} \right]. \]

\section{Proof of Lemma \ref{lemma:bockRiskZero}}\label{sec:proofBockRisk}

\begin{proof}
We begin by deriving the risk of $\bsdeltb$ at $\bstau = \boldsymbol 0$. Applying Lemma \ref{lemma:sure}, we obtain the risk expression 
\begin{equation}\label{eq:bockRisk}
\begin{aligned}
\mathcal{R}\left( \bsdeltb, \bstau \right) &= \Tr(\boldsymbol \Sigma) + \E \left( \frac{(\tilde p - 2)^2 ||\hbstau||_2^2}{(\hbstau^\tran \bsSig^{-1} \hbstau)^2}  - 
 \frac{2(\tilde p - 2)}{\hbstau^\tran \bsSig^{-1} \hbstau} \Tr(\bsSig) + \frac{4(\tilde p - 2)|| \hbstau ||_2^2}{(\hbstau^\tran \bsSig^{-1} \hbstau)^2} \right)  \\
&= \Tr(\boldsymbol \Sigma) + (\tilde p^2 - 4) \E \left( \frac{||\hbstau||_2^2}{(\hbstau^\tran \bsSig^{-1} \hbstau)^2} \right) - 2(\tilde p - 2)\Tr(\boldsymbol \Sigma) \E \left( \frac{1 }{\hbstau^\top \bsSig^{-1} \hbstau}  \right)  .
\end{aligned}
\end{equation}
If $||\bstau||_2^2 = 0$, then $\hbstau \sim \mathcal{N}(\boldsymbol 0, \bsSig)$. Define $\bsZ=\bsSig^{-1/2}\hbstau$, so $\bsZ\sim \mathcal N(\boldsymbol 0, \boldsymbol I)$. Then
\[ (\hbstau^\top \bsSig^{-1} \hbstau) = || \bsZ ||_2^2 \sim \chi^2_K. \] 
The inverse-$\chi^2$ moment yields 
\[ \e \left(\frac{1}{\hbstau^\top \bsSig^{-1}\hbstau}\right) =\frac{1}{K-2}.\]
To evaluate the remaining term, observe
\[ \frac{\|\hbstau\|_2^2}{(\hbstau^\top\bsSig^{-1}\hbstau)^2} = \frac{\bsZ^\top \bsSig \bsZ}{||\bsZ||_2^4} = \frac{||\bsZ||_2^2 \cdot \bsU^\top \bsSig \bsU}{||\bsZ||_2^4} = \frac{\bsU^\top \bsSig \bsU}{||\bsZ||_2^2} \]
where $\bsU=\bsZ/||\bsZ||_2$, such that $\bsU$ is uniform on the unit sphere and independent of $||\bsZ||_2$.
By symmetry, $\e(\bsU^\top\bsSig \bsU)=\Tr(\bsSig)/K$. Applying independence and the prior result, we obtain
\[ \e \left(\frac{\|\hbstau\|_2^2}{(\hbstau^\top\bsSig^{-1}\hbstau)^2}\right) = 
\frac{\mathrm{tr}(\bsSig)}{K(K-2)}.\]
Plugging these identities into Equation \ref{eq:bockRisk} yields 
\begin{equation}\label{eq:bockApproxRisk}
\mathcal{R}\left(\bsdeltb,\boldsymbol 0\right) = \Tr(\bsSig)\left( 1 -\frac{2(\tilde p-2)}{K} +\frac{(\tilde p-2)^2}{K(K-2)} \right),
\end{equation}
so the regularization term evaluates to 
\[ \Delta=\Tr(\bsSig)h(\tilde p), \hspace{5mm} \text{ where } \hspace{5mm} h(p)=-\frac{2(p-2)}{K}+\frac{(p-2)^2}{K(K-2)}. \]

Next, we find the gradient of this expression with respect to the sample sizes at the Neyman Allocation. By the Chain Rule, for $j = 0, \dots, K$,
\[ \frac{\partial\Delta}{\partial n_j}=h(\tilde p)\frac{\partial\Tr(\bsSig)}{\partial n_j}+\Tr(\bsSig)h'(\tilde p)\frac{\partial\tilde p}{\partial n_j}, \hspace{5mm} \text{ where } \hspace{5mm} h'(p)=\frac{2(p-K)}{K(K-2)}. \]
Denote as $\bsSigmaN$ the covariance matrix of $\hbstau$ at the Neyman allocation, which is given by 
\begin{equation}\label{eq:sigmaneyman}
\bsSigmaN = \sna \left( \text{diag}\left( \sqrt{V_1}, \dots, \sqrt{V_K} \right) + \sqrt{\frac{V_0}{K}} \boldsymbol{1} \boldsymbol{1}^\top\right), 
\end{equation}
where
\[ \sna = \frac{1}{N} \left( \sqrt{K V_0} + \sum_{k = 1}^K \sqrt{V_k} \right).\] 
At the Neyman allocation, for $k \in \{1, \dots, K\}$, 
\[ \left.\frac{\partial\Tr(\bsSig)}{\partial n_0}\right|_{\bsn=\bsn^\star}=-\frac{KV_0}{(n_0^\star)^2}=-\sna^2 \hspace{5mm} \text{ and } \hspace{5mm} \left.\frac{\partial\Tr(\bsSig)}{\partial n_k}\right|_{\bsn=\bsn^\star}=-\frac{V_k}{(n_k^\star)^2}=-\sna^2, \]
Also denote as $\bsnu = (\nu_1, \dots, \nu_K)^\top$ the dominant eigenvector associated with $\lambdamaxn$, scaled such that $||\bsnu||_2^2 = 1$. Then

\[ \left.\frac{\partial\tilde p}{\partial n_0}\right|_{\bsn=\bsn^\star}=\sna^2\frac{\frac{\Tr(\bsSigmaN)}{K}(\mathbbm 1^\tran\bsnu)^2-\lambdamaxn}{\lambdamaxn^2} \]
and
\[ \left.\frac{\partial\tilde p}{\partial n_k}\right|_{\bsn=\bsn^\star}=\sna^2\frac{\Tr(\bsSigmaN)\nu_k^2-\lambdamaxn}{\lambdamaxn^2}, \hspace{5mm} k=1,\dots,K, \]
where we have used the fact that, at the Neyman allocation, 
\[ \frac{K V_0}{(n_0^\star)^2} = \frac{V_1}{(n_1^\star)^2} = \dots = \frac{V_K}{(n_K^\star)^2} = \sna^2\]
(see the Appendix, Section \ref{sec:proof_neymanAlloc} for more details).

Plugging into the Chain Rule, we obtain
\[ \left.\frac{\partial\Delta}{\partial n_0}\right|_{\bsn=\bsn^\star}=\sna^2\left[-h(\tilde p^\star)+\frac{\Tr(\bsSigmaN)h'(\tilde p^\star)}{\lambdamaxn^2}\left\{\frac{\Tr(\bsSigmaN)}{K}(\mathbbm 1^\tran\bsnu)^2-\lambdamaxn\right\}\right] \]
and
\[ \left.\frac{\partial\Delta}{\partial n_k}\right|_{\bsn=\bsn^\star}=\sna^2\left[-h(\tilde p^\star)+\frac{\Tr(\bsSigmaN)h'(\tilde p^\star)}{\lambdamaxn^2}\left\{\Tr(\bsSigmaN)\nu_k^2-\lambdamaxn\right\}\right], \hspace{5mm} k=1,\dots,K. \]
where
\[ \tilde p^\star=\frac{\Tr(\bsSigmaN)}{\lambdamaxn}. \]

The coefficient multiplying the bracketed terms must be negative because $\tilde p^\star < K$ and thus $h'(\tilde p^\star) < 0$. Thus, we see that the derivative is strictly decreasing in $\nu_k^2$. By the Bunch-Nielsen-Sorensen formula \citep{bunch1978rank}, $\bsnu$ will have $k^{th}$ entry such that 
\[\nu_k \propto \frac{1}{\lambdamaxn/\sna-\sqrt{V_k}}.\] 
Hence, larger values of $V_k$ will correspond to larger values of $\nu_k^2$. This establishes the claimed ordering of the derivatives among the active treatment arms.

Next, we invoke one additional claim.
\begin{claim}\label{claim:bockMinAlloc}
Under the condition
\[ \left(\max_{k\geq1}\sqrt{V_k}-\min_{k\geq1}\sqrt{V_k} \right) \leq \frac{1}{2}\sqrt{KV_0}, \]
the dominant eigenvector satisfies
\[ \frac{(\mathbbm 1^\tran\bsnu)^2}{K}>\max_{k\geq1}\nu_k^2. \]
\end{claim}
\noindent The proof is provided in Appendix Section~\ref{sec:proof_bockMinAlloc}.

Now, observe that the two derivative expressions differ only through the bracketed terms, and that the common pre-factor for the bracketed terms must be negative because $\tilde p^\star < K$ and thus $h'(\tilde p^\star) < 0$. By Claim~\ref{claim:bockMinAlloc}, the bracketed term is largest for the control arm, and thus 
\[ \left.\frac{\partial\Delta}{\partial n_0}\right|_{\bsn=\bsn^\star}<\min_{k\geq1}\left.\frac{\partial\Delta}{\partial n_k}\right|_{\bsn=\bsn^\star} \]
when $\bstau=\boldsymbol 0$.

Lastly, the derivatives of the general regularization term $\Delta$
with respect to $\bsn$ are continuous in $\bstau$ in a neighborhood of
$\boldsymbol 0$. Because the inequalities are strict at $\bstau=\boldsymbol 0$,
there exists an $\epsilon>0$ such that the same derivative ordering holds whenever
\[ \bstau^\tran\bsSigmaN^{-1}\bstau < \epsilon.\]
\end{proof}

\section{Proof of Claim \ref{claim:bockMinAlloc}}\label{sec:proof_bockMinAlloc}

\begin{proof}

Under the heteroscedasticity condition on the potential outcomes (Condition \ref{eq:riskConditionBock}) we have
\[ \left( \max_{k \geq 1}\sqrt{V_k}-\min_{k \geq 1}\sqrt{V_k} \right) \leq \frac{1}{2} \sqrt{K V_0}. \]
By the Bunch-Nielsen-Sorensen formula \citep{bunch1978rank}, the dominant eigenvector $\bsnu$ will have $k^{th}$ entry proportional to $\left(\lambdamaxn/\sna-\sqrt{V_k}\right)^{-1}$. Condition \ref{eq:riskConditionBock} induces a simple bound on the ratio between the largest and smallest entries of $\bsnu$:
\begin{align*}
\frac{\max_k \nu_k}{\min_k \nu_k} &= \frac{\lambdamaxn/\sna - \min_{k \geq 1} \sqrt{V_k}}{\lambdamaxn/\sna - \max_{k \geq 1} \sqrt{V_k}} = 1 + \frac{\max_{k \geq 1} \sqrt{V_k} - \min_{k \geq 1} \sqrt{V_k}}{\lambdamaxn/\sna - \max_{k \geq 1} \sqrt{V_k}} \\
& \leq 1 + \frac{\frac{1}{2} \sqrt{K V_0}}{\sqrt{K V_0} + \left( \frac{1}{K} \sum_{k = 1}^K \sqrt{V_k} - \max_{k \geq 1 } \sqrt{V_k}\right)} \\
& \leq 1 + \frac{\frac{1}{2} \sqrt{K V_0}}{\sqrt{K V_0} + \left( \min_{k \geq 1} \sqrt{V_k} - \max_{k \geq 1 } \sqrt{V_k}\right)} \\
& \leq 1 + \frac{\frac{1}{2} \sqrt{K V_0}}{\sqrt{KV_0} - \frac{1}{2} \sqrt{K V_0}} = 2.
\end{align*}

If the entries of the dominant eigenvector $\bsnu$ cannot vary by more than a factor of 2, this implies:
\begin{align*}
 \frac{(\mathbbm{1}^T \bsnu)^2}{K} &\geq \frac{1}{K} \left((K - 1) \min_k \nu_k + \max_k \nu_k \right)^2 \\
 &\geq \frac{1}{K} \left(\frac{K - 1}{2} \max_k \nu_k + \max_k \nu_k \right)^2 \\
 &= \frac{1}{K} \left(\frac{K + 1}{2} \max_k \nu_k \right)^2 \\
 &= \frac{(K + 1)^2}{4K} \max_k \nu_k^2\\
 & > \max_k \nu_k^2,
\end{align*}
where the last line follows from the fact that the minimizing value of $(K + 1)^2/4K$ over $K \geq 2$ positive integers is $1.125$. Hence, this directly implies
\[ \frac{(\mathbbm 1^\tran\bsnu)^2}{K}>\max_{k\geq1}\nu_k^2. \]
\end{proof}

\section{Heuristic risk approximations for $\bsdelts$}\label{sec:riskApprox}

The SURE-minimizing estimator $\bsdelts$ does not admit a closed-form risk expression when $\bstau = 0$. Hence, we rely on a sequence of approximations to characterize the behavior of risk-minimizing allocations.

Applying Lemma~\ref{lemma:sure} yields
\begin{equation}\label{eq:riskStatementSteinMin}
\begin{aligned}
\mathcal{R}\left( \bsdelts, \bstau \right)
&=
\Tr(\bsSigma)
+
4 \Tr(\bsSig)\,
\e \left( \frac{\hbstau^\tran \bsSig \hbstau}{\|\hbstau\|_2^4} \right)
-
\Tr(\bsSig)^2\,
\e \left( \frac{1}{\|\hbstau\|_2^2} \right).
\end{aligned}
\end{equation}

To obtain a usable approximation, we approximate the Gaussian quadratic forms
appearing in Equation~\eqref{eq:riskStatementSteinMin} using the moment-matching
approach of \cite{satterthwaite1946approximate}. Namely, we approximate
\[
\|\hbstau\|_2^2 \sim \text{Gamma}\!\left( \frac{m^2}{v}, \frac{v}{m} \right),
\]
where
\[
m := \e\!\left( \|\hbstau\|_2^2 \right)
= \Tr(\bsSig) + \|\bstau\|_2^2,
\qquad
v := \var\!\left( \|\hbstau\|_2^2 \right)
= 2\Tr(\bsSig^2) + 4\,\bstau^\tran \bsSig \bstau .
\]
Via this approximation, we obtain a tractable expression for
\[
\e \left( \frac{1}{\|\hbstau\|_2^2} \right)
\;\approx\;
\frac{\Tr(\bsSig) + \|\bstau\|_2^2}
{\big(\Tr(\bsSig) + \|\bstau\|_2^2\big)^2
- 2\Tr(\bsSig^2) - 4\bstau^\tran \bsSig \bstau}.
\]

To approximate the ratio expectation, we employ the crude
approximation
\[
\hbstau^\tran \bsSig \hbstau
\;\approx\;
\frac{\e\!\left( \hbstau^\tran \bsSig \hbstau \right)}
{\e\!\left( \|\hbstau\|_2^2 \right)}\,
\|\hbstau\|_2^2,
\]
which holds exactly when $\bsSig$ is proportional to the identity. Using
\[
\e\!\left( \hbstau^\tran \bsSig \hbstau \right)
= \Tr(\bsSig^2) + \bstau^\tran \bsSig \bstau,
\]
this yields the approximation
\[
\e \left( \frac{\hbstau^\tran \bsSig \hbstau}{\|\hbstau\|_2^4} \right)
\;\approx\;
\frac{\Tr(\bsSig^2) + \bstau^\tran \bsSig \bstau}
{\Tr(\bsSig) + \|\bstau\|_2^2}\,
\e \left( \frac{1}{\|\hbstau\|_2^2} \right).
\]

Substituting these approximations into
Equation~\ref{eq:riskStatementSteinMin} yields a closed-form approximation to the
risk. In the low signal-to-noise regime, we treat $\|\bstau\|_2^2$ and
$\bstau^\tran \bsSig \bstau$ as small relative to $\Tr(\bsSig)$ and
$\Tr(\bsSig^2)$, and retain only the leading terms obtained by setting
$\bstau=\boldsymbol 0$ inside the approximations. This yields
\[ \mathcal R(\bsdelts,\bstau) \approx \Tr(\bsSig) \left( 1 + \frac{4\Tr(\bsSig^2)-\Tr(\bsSig)^2}{\Tr(\bsSig)^2-2\Tr(\bsSig^2)} \right).\]

\section{Simulation design}\label{sec:sim_design}

An experiment is a fixed regime of simulation parameters. Within each experiment, run multiple iterations, where each iteration has a particular set of fixed potential outcomes. Assume the following data-generating process:

\begin{itemize}
    \item $\mu_0$ is the mean of the control treatment arm.
    \item $\bstau$ is the vector of treatment effects, $\{\tau_1, \dots \tau_K\}$.
    \item Then $\bsmu = \mu_0 + \bstau = \begin{pmatrix} \mu_0 \\ \mu_0 + \tau_1 \\ \vdots \\ \mu_0 + \tau_{K} \\ \end{pmatrix}$ is the vector of potential outcome means.
    \item $V_0$ is the variance of the control treatment arm, and $V_1, \dots, V_K$ are the variances of the active treatment arms.
    \item $\bsOmega$ is a $K \times K$ diagonal matrix and is the true covariance of the potential outcomes, and is a function of the variances $V_0, \dots, V_K$. We assume $\cov(\mathbf{Y}_i, \mathbf{Y}_j) = 0$.
    \item Define $\bsY_i = \{Y_i(0), Y_i(1), \dots Y_i(K)\}$ as the vector of potential outcomes for unit $i$.
\end{itemize}

Then the potential outcomes for each unit $\bsY_i$ are generated:
\[ \mathbf{Y}_i \sim \mathcal{N}_{K + 1}\left(\bsmu,\bsOmega \right).\]

For each experiment, set fixed the data-generating process for the potential outcomes. The DGP is based on whether $\bstau$ is sparse or dense, the control-variance regime, and the signal-to-noise ratio $\kappa$.

\begin{itemize}
    \item \texttt{K}: number of \emph{active} treatments. Total number of arms = $K + 1$ ($K$ active arms + 1 control arm).
    \item  \texttt{Omega}: true (diagonal) covariance matrix of potential outcomes. Active treatment arm variances are sampled iid from a log normal distribution with log-location parameter of $\log(350)$ and log-standard-deviation parameter $0.60$. Then:
    \begin{itemize}
        \item if in the low control-variance regime, $V_0$ is set to half of the average potential outcome variance of the treatment arm, i.e. $V_0 = \tfrac{1}{2}\cdot K^{-1}\sum_{k=1}^K V_k$.
        \item if in the high control-variance regime, $V_0$ is set to four times the average variance of the treatment arm, i.e. $V_0 = 4\cdot K^{-1}\sum_{k=1}^K V_k$.
    \end{itemize} 
    \item \texttt{tau}: the $K$-vector of treatment effects.  
    \begin{itemize}
        \item if in the ``dense" setting, an unscaled $\tilde{\bstau}$ is sampled i.i.d. from a $\text{Unif}([1, 2])$ distribution, resulting in all causal effects having comparable magnitude. 
        \item if in the ``sparse" setting, $\tilde{\bstau} = (1, \dots, 0)$ has only one nonzero entry, which corresponds to the first arm.
        \item effects are then scaled to match the desired value of $\kappa$, the signal-to-noise ratio. 
        \item \texttt{mu0} is set to $0$ so that there is a fixed mapping from \texttt{tau} to \texttt{mu} (the $K+1$ vector of arm means).
    \end{itemize}
    \item \texttt{N}: number of observations in one iteration.
\end{itemize}

Also set fixed the sequential experimental design procedure:
\begin{itemize}
    \item \texttt{b1}: initial batch size. How many observations are used to get initial estimates of mean and variance before risk minimization or other adaptive algorithm starts up. Set to 100, a rough upper bound to the $\sqrt{N}$ suggestion from \cite{zhao2023adaptive}.
    \item Choose a risk function, from the following options: SURE-min, Bock, Dimmery, Neyman. Alternatively, you can follow complete randomization, where instead of minimizing the risk at each step, you just sequentially iterate through the different treatment arms. The Neyman risk minimizes the variance by minimizing the sum of the diagonal of the covariance matrix of the treatment effects, $\hat{\bsSig}_i$ (see below).
    \item Choose a list of estimators, from the following options: difference in means, SURE-min, Bock, Dimmery.
\end{itemize}

Finally, fix \texttt{n.iters}, the number of iterations for the experiment.

Note that there are two components that can be chosen independently: the risk function, and the estimator. For example, you could do sequential assignment using the Bock risk function, but then after the experiment is over, you could estimate the treatment effects without shrinkage using the difference-in-means estimator. 

Within each iteration:
\begin{itemize}
    \item Generate \texttt{Y}, the matrix of potential outcomes, a $N \times (K + 1)$ matrix.
    \item At each step $i = 1, \dots N$
    \begin{itemize}
        \item If $i \leq b_1$, then assign $W_i$ sequentially.
        \item If $i > b_1$, then:
        \begin{itemize}
            \item Estimate the running mean of each arm, $\hat{\mu}_{k, i}$ for $k = 0, \dots, K$ (e.g. running mean includes observations up to $i-1$).
            \item Estimate the running variance of each arm, $\hat{V}_{k, i}$ for $k = 0, \dots, K$.
            \item Estimate the running treatment effects using the difference-in-means estimator, $\hbstau_{i}$.
            \item Estimate the running covariance matrix of the treatment effects, $\hat{\bsSig}_i$.
            \item Using the chosen risk function for the experiment, calculate what the estimated risk would be if the next unit were to be assigned to each arm. Define $\hat{\mathcal{R}}_{k,i}$ as the risk if the next unit $i$ is assigned to arm $k$.  Recall that the risk $\hat{\mathcal{R}}_{k,i}$ is the average expected loss over all arms.  Calculate $\hat{\mathcal{R}}_{k,i}$ for $k = 0, \dots, K$ (using $\hbstau_i$ and $\hat{\bsSig}_i$). 
            \item Choose the arm with the smallest risk $\hat{\mathcal{R}}_{k,i}$ and assign $W_i = \underset{k \in 0, \dots, K}{\operatorname{argmin}} \ \hat{\mathcal{R}}_{k,i}$.
            \item Record the outcome $Y_i^{obs} = Y_i(W_i)$.
        \end{itemize}
    \end{itemize}
    \item At the end of the iteration, we now have:
        \begin{itemize}
        \item An N-vector of observed outcomes $\bsY^{obs}$.
        \item An N-vector of allocated treatments $\bsW$.
        \item A list of N values of $\hbstau_{i}$ and $\hat{\bsSig}_i$.
        \end{itemize}
\end{itemize}

To evaluate the experiment:
\begin{itemize}
    \item First evaluate each iteration:
    \begin{itemize}
        \item For each step $i$, compute the vector of $K$ MSE values of each estimator. For example, the MSE for $\bsdeltb_i$ is $(\hat{\delta}_{B, k, i}- \tau_k)^2$ for $k = 1, \dots K$.
        \item For each step $i$, calculate the ``compound'' MSE, which is the sum of the $K$ MSE values.
        \item Track trajectories of compound MSE as units arrive.
    \end{itemize}
    \item Across iterations:
    \begin{itemize}
        \item Calculate the ``mean trajectory'' by averaging the compound MSE values across iterations, for each unit.
        \item We now have the average compound MSE for each estimator for each unit.
        \item Then we plot the trajectories for each estimator.
    \end{itemize}
\end{itemize}

Finally, we plot multiple experiments, with the same parameter DGP but different risk functions, on the same panel of plots, to evaluate which risk function and estimators produced the lowest MSE.

\section{Additional simulation results\label{sec:sim_add_results}}

\begin{table}[H]
\centering
\begin{threeparttable}
\caption{Summary of simulation performance of adaptive treatment allocation strategies for $K = 9$}
\label{tab:adaptiveSimsK9}

\small
\setlength{\tabcolsep}{4pt}

\begin{tabular}{rlll|rrrrr|rrrrr}
\toprule
 &  &  & 
& \multicolumn{5}{c}{MSE at Unit $1000$}
& \multicolumn{5}{c}{Mean MSE over Unit $\geq 1000$} \\
\cmidrule(lr){5-9} \cmidrule(lr){10-14}
$K$& $V_0$ & $\bstau$ & $\kappa$ & CR & Neyman & Bock & SURE-min & Dimmery
& CR & Neyman & Bock & SURE-min & Dimmery \\
\midrule
9 & high & dense & 0 & 102.0 & 61.5 & 55.1 & 6.0 & \textbf{4.7} & 46.2 & 28.1 & 25.3 & 2.4 & \textbf{2.1}\\
9 & high & sparse & 0 & 90.6 & 60.8 & 57.8 & 4.2 & \textbf{3.4} & 40.5 & 27.6 & 26.7 & 1.9 & \textbf{1.5}\\
9 & high & dense & 2 & 103.8 & 65.5 & 58.4 & 13.6 & \textbf{13.1} & 46.8 & 30.2 & 26.5 & 10.2 & \textbf{10.1}\\
9 & high & sparse & 2 & 85.3 & 74.4 & 51.3 & 7.2 & \textbf{7.0} & 39.9 & 34.8 & 23.2 & 4.5 & \textbf{4.1}\\
9 & high & dense & 4 & 94.1 & 56.9 & 56.3 & \textbf{18.5} & 20.8 & 44.9 & 26.3 & 25.9 & \textbf{13.3} & 16.3\\
9 & high & sparse & 4 & 86.9 & 44.2 & 42.5 & 8.1 & \textbf{6.8} & 41.4 & 20.6 & 19.0 & 5.6 & \textbf{5.1}\\
\midrule
9 & low & dense & 0 & 49.7 & 40.7 & 31.8 & 2.6 & \textbf{1.8} & 22.8 & 18.5 & 13.8 & 1.2 & \textbf{0.8}\\
9 & low & sparse & 0 & 47.3 & 39.8 & 28.7 & \textbf{1.9} & 2.0 & 21.3 & 18.2 & 12.1 & 0.9 & \textbf{0.9}\\
9 & low & dense & 2 & 49.6 & 34.9 & 25.9 & \textbf{5.2} & 5.9 & 22.6 & 16.3 & 11.9 & \textbf{3.9} & 4.8\\
9 & low & sparse & 2 & 40.6 & 47.9 & 27.5 & 3.7 & \textbf{3.6} & 18.7 & 21.5 & 12.4 & \textbf{2.3} & 2.4\\
9 & low & dense & 4 & 37.1 & 30.9 & 24.1 & 8.1 & \textbf{7.9} & 16.7 & 14.3 & 11.4 & \textbf{6.3} & 6.5\\
9 & low & sparse & 4 & 48.0 & 43.2 & 28.2 & 6.3 & \textbf{5.7} & 22.0 & 19.6 & 13.1 & 4.9 & \textbf{4.1}\\
\bottomrule
\end{tabular}
\end{threeparttable}
\end{table}

\begin{figure}[H]
    \centering
    \includegraphics[width=0.9\linewidth]{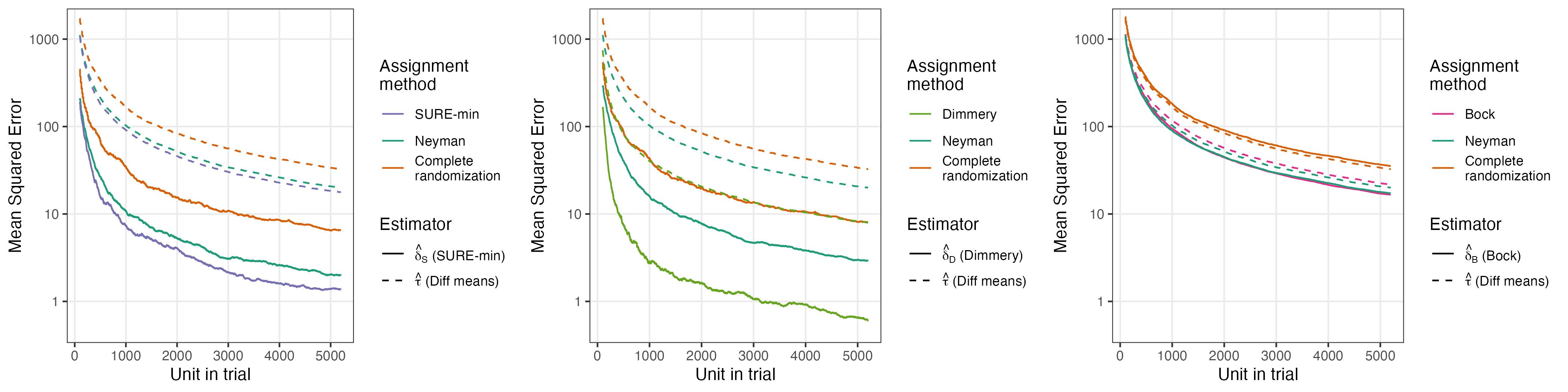}
    \caption{$K = 12, \kappa=0$.}
    \label{fig:K12-kappa0}
    \includegraphics[width=0.9\linewidth]{figures/rm_regime_K-6_V-high_Kappa-2_Tau-sparse_panel.png}
    \caption{$K = 12, \bstau$ sparse, $\kappa=2$.}
    \label{fig:K12-sparse-kappa2}
    \includegraphics[width=0.9\linewidth]{figures/rm_regime_K-6_V-high_Kappa-2_Tau-dense_panel.png}
    \caption{$K = 12, \bstau$ dense, $\kappa=2$.}
    \label{fig:K12-dense-kappa2}
\end{figure}

\begin{figure}[H]
    \centering
    \includegraphics[width=0.9\linewidth]{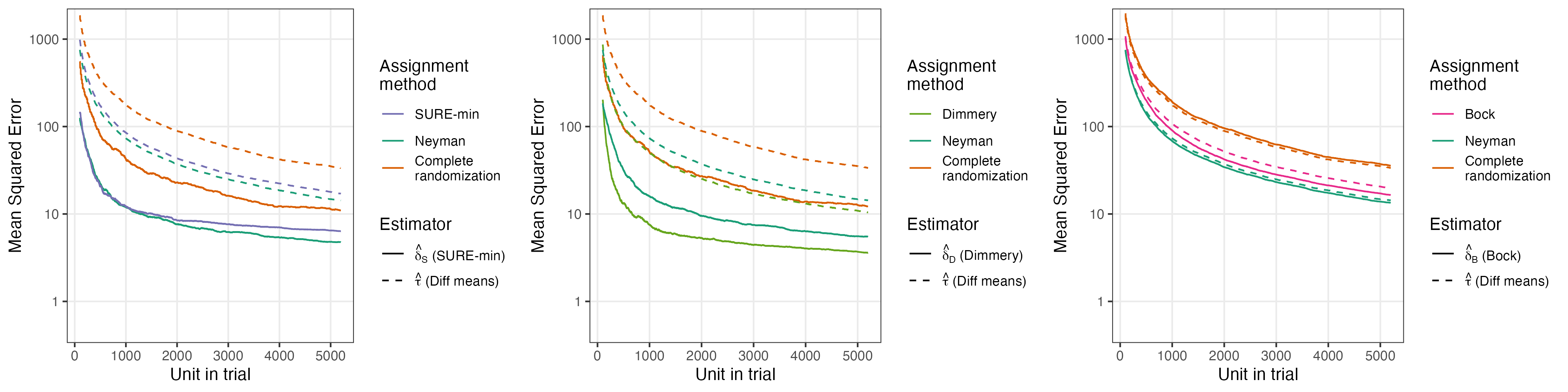}
    \caption{$K = 12, \bstau$ sparse, $\kappa=4$.}
    \label{fig:K12-sparse-kappa4}
    \includegraphics[width=0.9\linewidth]{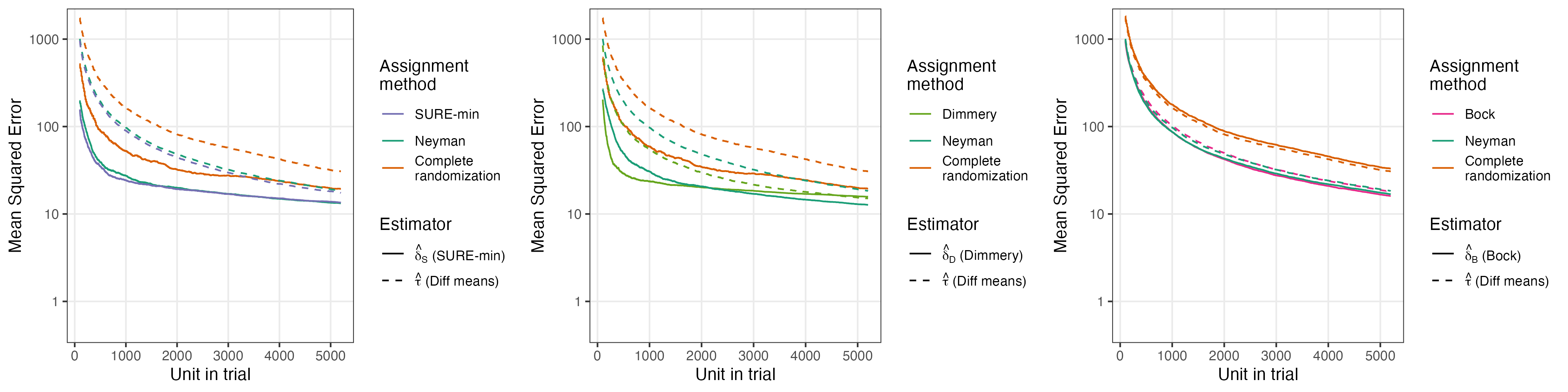}
    \caption{$K = 12, \bstau$ dense, $\kappa=4$.}
    \label{fig:K12-dense-kappa4}
\end{figure}

\begin{figure}[H]
    \centering
    \includegraphics[width=0.9\linewidth]{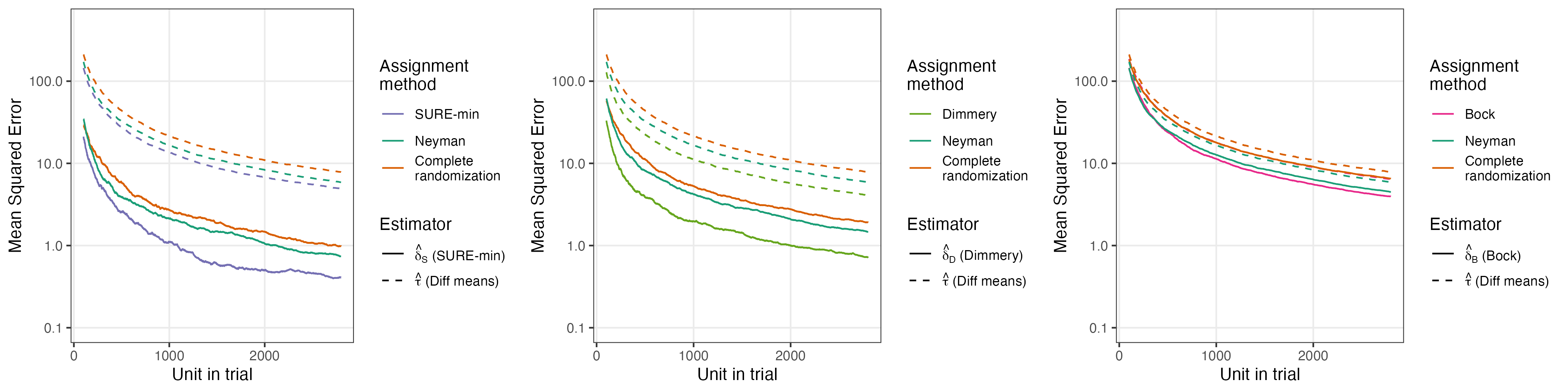}
    \caption{$K = 6$, low-variance control regime, $\kappa=0$.}
    \label{fig:K6-low-kappa0}
    \includegraphics[width=0.9\linewidth]{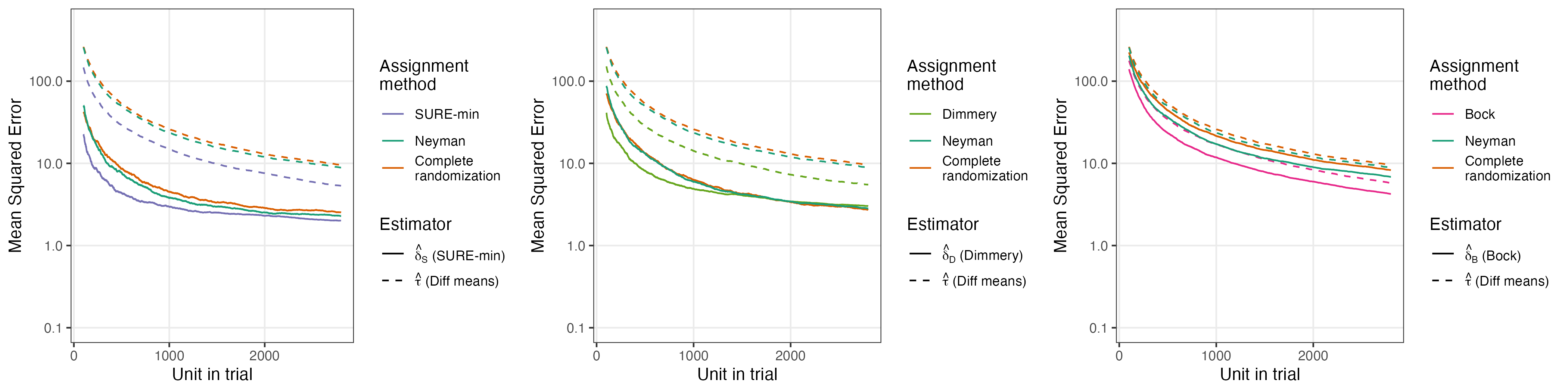}
    \caption{$K = 6$, low-variance control regime, $\bstau$ sparse, $\kappa=2$.}
    \label{fig:K6-low-sparse-kappa2}
    \includegraphics[width=0.9\linewidth]{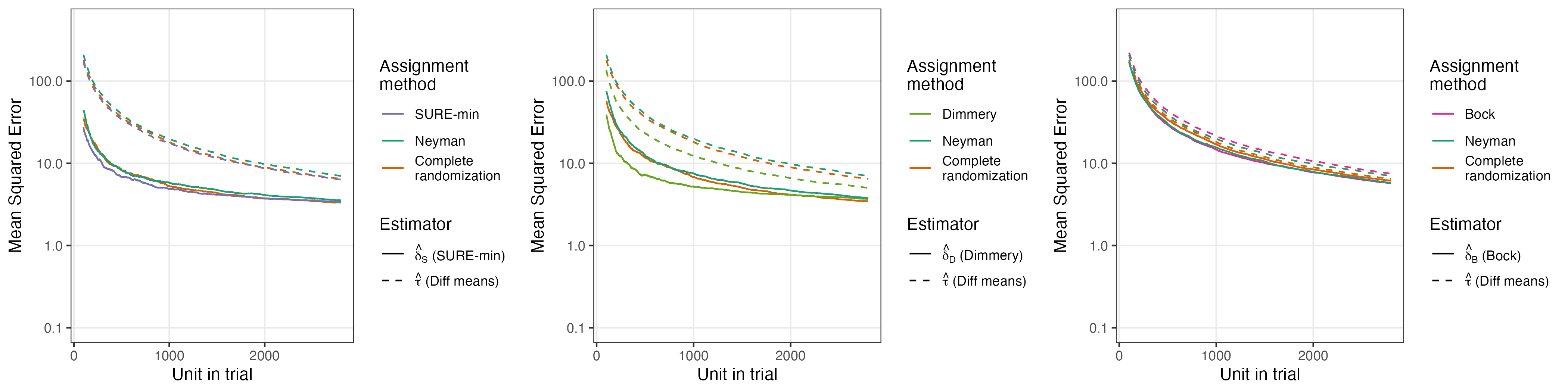}
    \caption{$K = 6$, low-variance control regime, $\bstau$ dense, $\kappa = 2$.}
    \label{fig:K6-low-dense-kappa2}
\end{figure}

\begin{figure}[H]
    \centering
    \includegraphics[width=0.9\linewidth]{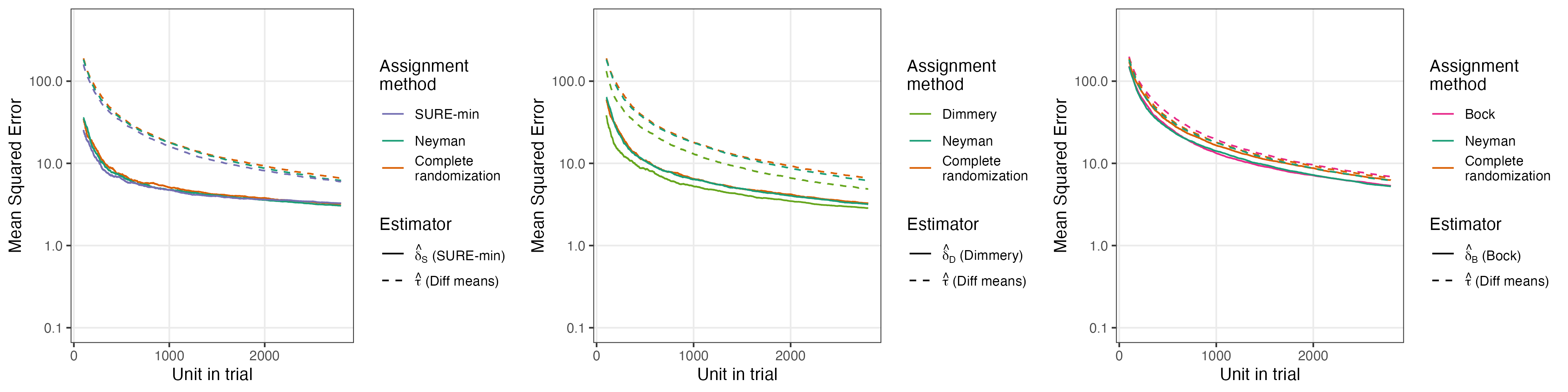}
    \caption{$K = 6$, low-variance control regime, $\bstau$ sparse, $\kappa = 4$.}
    \label{fig:K6-low-sparse-kappa4}
    \includegraphics[width=0.9\linewidth]{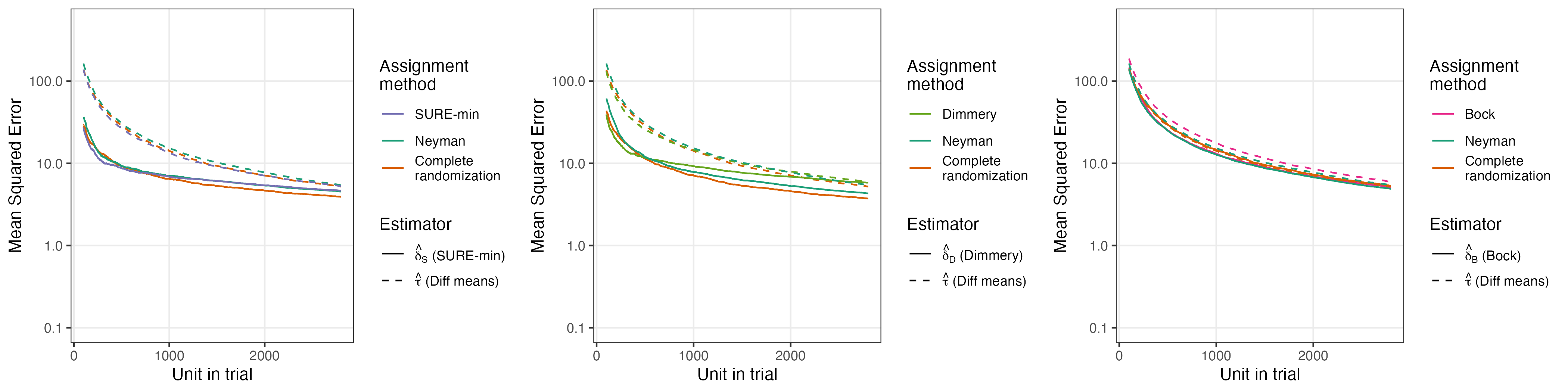}
    \caption{$K = 6$, low-variance control regime, $\bstau$ dense, $\kappa = 4$.}
    \label{fig:K6-low-dense-kappa4}
\end{figure}

\begin{figure}[H]
    \centering
    \includegraphics[width=0.9\linewidth]{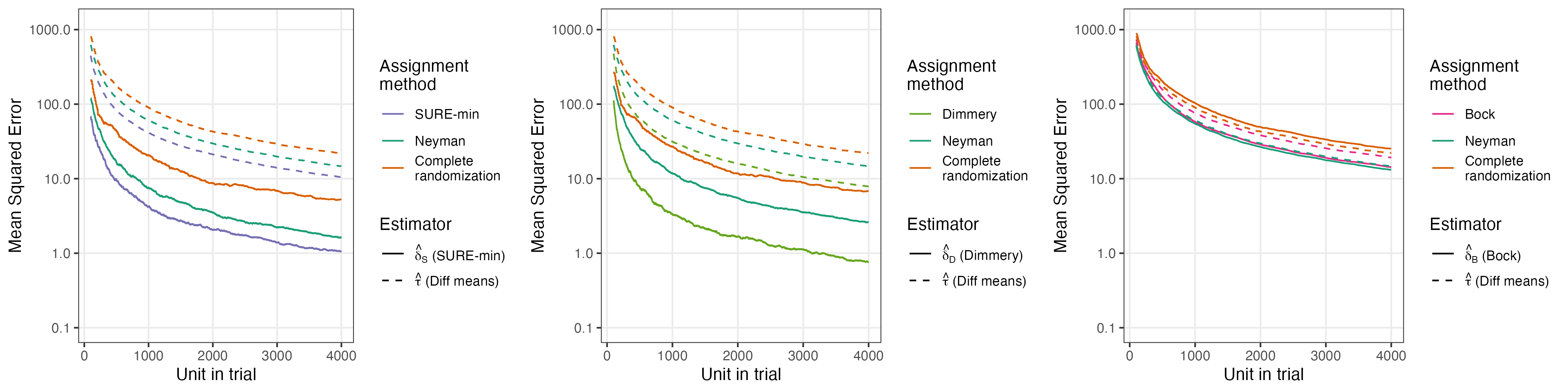}
    \caption{$K = 9$, high-variance control regime, $\kappa=0$.}
    \label{fig:K9-kappa0}
    \includegraphics[width=0.9\linewidth]{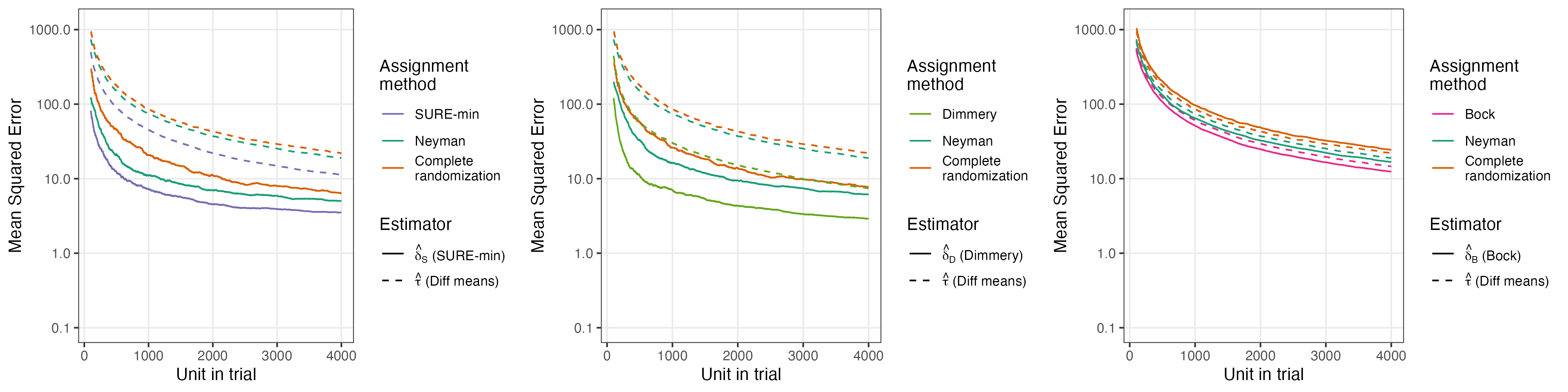}
    \caption{$K = 9$, high-variance control regime, $\bstau$ sparse, $\kappa=2$.}
    \label{fig:K9-sparse-kappa2}
    \includegraphics[width=0.9\linewidth]{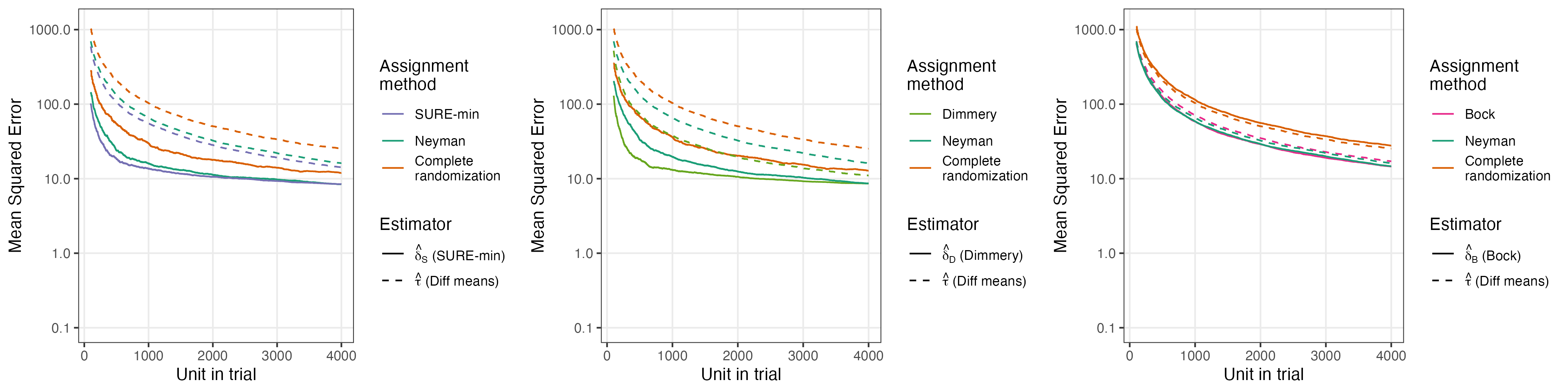}
    \caption{$K = 9$, high-variance control regime, $\bstau$ dense, $\kappa = 2$.}
    \label{fig:K9-dense-kappa2}
\end{figure}

\begin{figure}[H]
    \centering
    \includegraphics[width=0.9\linewidth]{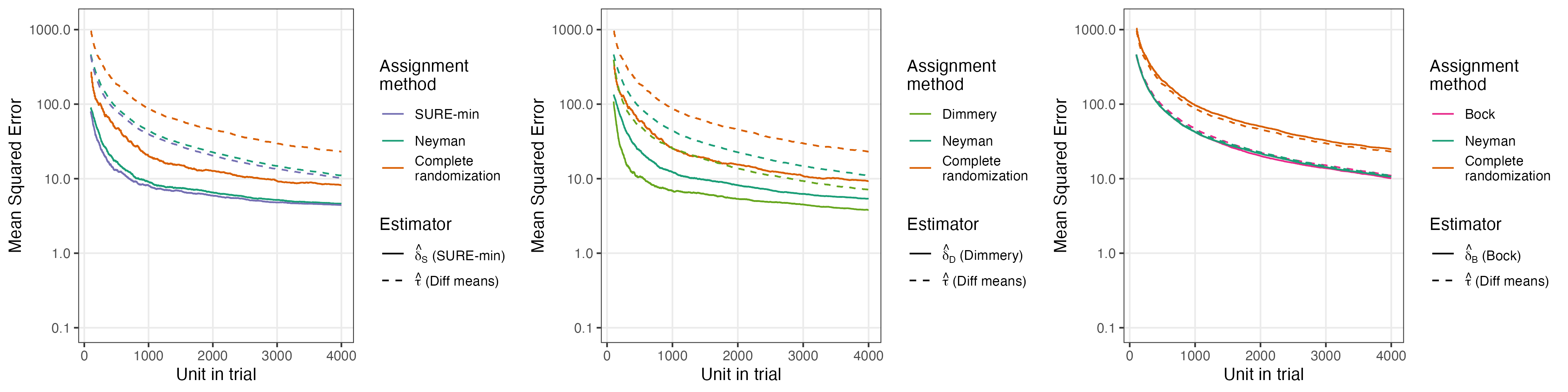}
    \caption{$K = 9$, high-variance control regime, $\bstau$ sparse, $\kappa = 4$.}
    \label{fig:K9-sparse-kappa4}
    \includegraphics[width=0.9\linewidth]{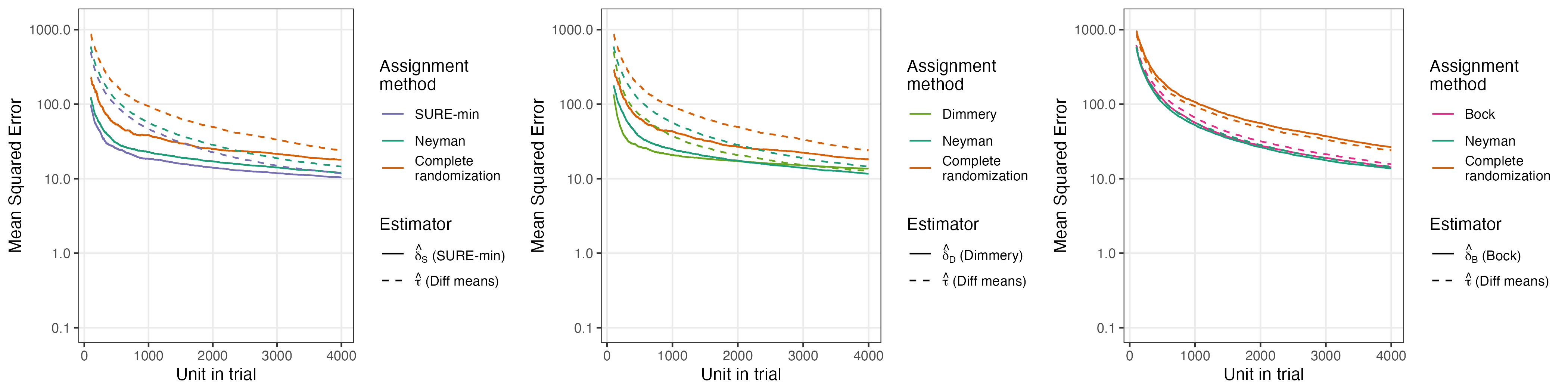}
    \caption{$K = 9$, high-variance control regime, $\bstau$ dense, $\kappa = 4$.}
    \label{fig:K9-dense-kappa4}
\end{figure}

\section{FAB Interval Computation and Additional Simulations}\label{appendix:fab}

To update our FAB intervals, we begin by considering the ``inferred" prior for Bock's estimator $\bsdeltb$ and the SURE-min estimator $\bsdelts$. Notably, both estimators shrink all arms by the same multiplicative amount. Consider the case in which 
\begin{equation}\label{eq:impliedModel}
\begin{aligned}
\bstau &\sim \mathcal{N}(\boldsymbol 0, \bspsi), \hspace{5mm} \text{ and } \\
\hbstau \mid \bstau & \sim \mathcal{N}(\bstau, \bsSig).
\end{aligned}
\end{equation}
Under this model, 
\[ \bstau \mid \hbstau \sim \mathcal{N}\!\left( \bspsi(\bspsi+\bsSig)^{-1}\hbstau,\; (\bspsi^{-1} + \bsSig^{-1})^{-1} \right). \]
Using the standard convention, Empirical Bayes estimators approximate the posterior mean, estimating unseen hyperparameters from the distribution of the data. Hence, the multiplicative structure of $\bsdeltb$ and $\bsdelts$ implies that 
\[ \bspsi(\bspsi + \bsSig)^{-1} = c \, \ident \]
for some constants $0 < c < 1$. Assuming $\bspsi$ and $\bsSig$ are full rank, this implies that $\bspsi = c/(1-c) \bsSig$ -- or, in other words, that $\bspsi = a \bsSig$ for some $a > 0$. Hence, this is the implied prior for our two estimators. 

Under this prior, we also have that 
\[ \hbstau \sim \mathcal{N} \left( \boldsymbol 0, (a + 1) \bsSig \right), \]
while our posterior mean can now be written as 
\[ \e \left( \bstau \mid \hbstau \right) = \frac{a}{a + 1} \hbstau.\]
There are several different moment-based estimators of $a$. We observe, for example, that 
\[ \frac{1}{K} \cdot \e \left( \hbstau^\top \bsSig^{-1} \hbstau \right) - 1 = a,\]
which suggests an estimator of the form
\[ \hat a = \bigg(\frac{1}{K} \cdot \left( \hbstau^\top \hbsSig^{-1} \hbstau \right) - 1\bigg)_+.  \]

Drawing on \cite{pratt1963shorter}, \cite{yu_hoff_2018} begin by considering the model $Y \sim \mathcal{N}(\theta, \sigma^2)$, with $\theta$ unknown and $\sigma^2$ known. For any choice of confidence level $0 < \alpha < 1$, they consider any interval of the form 
\[ \left( Y + \sigma z_{\alpha (1 - w(\theta))}, Y + \sigma z_{\alpha w(\theta)} \right),\]
where $z_\nu$ is the $\nu^{th}$ quantile of a standard normal distribution and $w: \mathbb{R} \to (0, 1)$ is a continuous, nondecreasing function. This interval can be shown to be a valid $\alpha$-level confidence set for $\theta$. Moreover, under the prior $\theta \sim \mathcal{N}(0, \eta^2)$, it can be shown that choosing
\[ w(\theta) = g^{-1} \left( \frac{2 \sigma \theta}{\eta^2} \right) \hspace{5mm} \text{ where } \hspace{5mm} g(w) = \Phi^{-1}(\alpha w) - \Phi^{-1}(\alpha(1 - w)) \]
yields the shortest interval in expectation with respect to the prior. In practice, this means that under the prior, one can solve for the optimal confidence interval endpoints by solving two nonlinear equations. In the multigroup setting with unknown variance, \cite{yu_hoff_2018} show that one can replace the prior by a model for heterogeneity across groups to attain shorter intervals on average.

In our case, the $j^{th}$ treatment effect is assumed to have prior variance equal to an $a$-multiple of the $j^{th}$ diagonal element of $\bsSig$. Hence, our function is instead 
\[ w(\theta) = g^{-1} \left(\frac{2 \theta }{a \sqrt{V_0/n_0 + V_j/n_j}}\right), \]
where $a$ can be estimated by $\hat a$ and the standard variance estimators can be plugged into the denominator. With these substitutions, we are able to compute the FAB interval endpoints. Additional simulation results for the $K = 12$ case are given in Table \ref{tab:k12-fab-high-hetero-high-v}. 

\begin{table}[ht]
\centering
\begin{tabular}{lllrrr}
\toprule
\makecell[l]{Signal-to-\\Noise Ratio} & $\bstau$ setting & \makecell[l]{Allocation \\Scheme} & \makecell[l]{Minimum \\ Wald Coverage} & \makecell[l]{Minimum \\ FAB Coverage} & \makecell[l]{Avg. FAB\\length reduction} \\
\midrule
$\kappa=0$ & $-$ & Balanced & 94\% & 95\% & $-$15\% \\
 & & Neyman & 94\% & 95\% & $-$15\% \\
 & & SURE-min & 94\% & 96\% & $-$15\% \\
 & & Dimmery & 96\% & 98\% & $-$16\% \\
 & & Bock & 93\% & 94\% & $-$15\% \\
\addlinespace
$\kappa=2$ & sparse & Balanced & 94\% & 95\% & $-$15\% \\
 & & Neyman & 94\% & 94\% & $-$15\% \\
 & & SURE-min & 95\% & 96\% & $-$15\% \\
 & & Dimmery & 96\% & 98\% & $-$15\% \\
 & & Bock & 93\% & 94\% & $-$15\% \\
\addlinespace
$\kappa=2$ & dense & Balanced & 93\% & 91\% & $-$14\% \\
 & & Neyman & 94\% & 93\% & $-$14\% \\
 & & SURE-min & 95\% & 94\% & $-$14\% \\
 & & Dimmery & 96\% & 96\% & $-$14\% \\
 & & Bock & 93\% & 93\% & $-$13\% \\
\addlinespace
$\kappa=4$ & sparse & Balanced & 94\% & 94\% & $-$15\% \\
 & & Neyman & 94\% & 94\% & $-$14\% \\
 & & SURE-min & 94\% & 94\% & $-$14\% \\
 & & Dimmery & 96\% & 96\% & $-$15\% \\
 & & Bock & 93\% & 94\% & $-$14\% \\
\addlinespace
$\kappa=4$ & dense & Balanced & 94\% & 93\% & $-$14\% \\
 & & Neyman & 94\% & 93\% & $-$12\% \\
 & & SURE-min & 94\% & 94\% & $-$13\% \\
 & & Dimmery & 96\% & 96\% & $-$14\% \\
 & & Bock & 93\% & 94\% & $-$12\% \\
\bottomrule
\end{tabular}
\caption{FAB interval performance for $K=12$ under high control variance. Coverage columns report the minimum across the $K$ treatment effects; length reduction is computed as the mean FAB interval length divided by the mean Wald interval length, minus one.}
\label{tab:k12-fab-high-hetero-high-v}
\end{table}

\end{document}